\title{Selective Disclosure in Overlapping Generations\footnote{We thank Eddie Dekel, Mira Frick, Drew Fudenberg, Ying Gao, Germ\'{a}n Giezewski, Peter Klibanoff, Suraj Malladi, Lara Sanchez, and Bruno Strulovici for helpful comments. Pei thanks the NSF Grant SES-2337566 for financial support.}}
\author{Nemanja Anti\'{c}\footnote{Northwestern University, \url{nemanja.antic@kellogg.northwestern.edu}} \and Harry Pei\footnote{Northwestern University, \url{ harrydp@northwestern.edu}}}
\date{\today}
\documentclass[pstricks,12pt]{article}
\usepackage{latexsym}
\usepackage{graphicx}
\usepackage{epsfig}
\usepackage{dsfont}
\usepackage{color}
\usepackage{natbib}
\usepackage{epstopdf}
\usepackage{hyperref}
\hypersetup{
    colorlinks=true,
    linkcolor=blue,
    citecolor=blue,
    filecolor=blue,
    urlcolor=blue
}
\usepackage{comment}
\usepackage{amsmath, amsthm, amssymb}
\usepackage{pstricks, pst-node}
\usepackage{tikz}
\usepackage{caption}
\usepackage{graphicx}
\graphicspath{ {./figures/} }
\usepackage{tikz}
\usetikzlibrary{decorations.pathreplacing, calc}
\usepackage{pgfplots}
\usepackage{subcaption}
\pgfplotsset{compat=1.18}
\pgfplotsset{
  every axis/.append style={font=\small}
}

\begin{document}
\numberwithin{equation}{section}

\maketitle
\begin{abstract}
    We develop an overlapping generations model where each agent observes a verifiable private signal about the state and, with positive probability, also receives signals disclosed by his predecessor. The agent then takes an action and decides which signals to pass on. Each agent's action has a positive externality on his predecessor and his optimal action increases in his belief about the state. 
We show that as the probability that messages reach the next generation approaches one, agents become increasingly selective in disclosing information. In the limit, all signals except for the most favorable ones will be concealed.\\

\end{abstract}

\newtheorem{Proposition}{\hskip\parindent\bf{Proposition}}
\newtheorem{Theorem}{\hskip\parindent\bf{Theorem}}
\newtheorem{Lemma}{\hskip\parindent\bf{Lemma}}
\newtheorem{Corollary}{\hskip\parindent\bf{Corollary}}
\newtheorem*{Definition}{\hskip\parindent\bf{Monotone Belief}}
\newtheorem*{Definition1}{\hskip\parindent\bf{Threshold Strategy}}
\newtheorem*{Definition2}{\hskip\parindent\bf{Constant-Threshold Strategy}}
\newtheorem*{Definition3}{\hskip\parindent\bf{Gradualism Property}}
\newtheorem*{Definition4}{\hskip\parindent\bf{Improvement Principle}}
\newtheorem*{Definition5}{\hskip\parindent\bf{No-Turning-Back Property}}
\newtheorem{Assumption}{\hskip\parindent\bf{Assumption}}
\newtheorem{Condition}{\hskip\parindent\bf{Condition}}
\newtheorem{Claim}{\hskip\parindent\bf{Claim}}

\newpage
\section{Introduction}\label{sec1}
Research shows that the intergenerational transmission of anecdotes and stories can shape individuals’ preferences, beliefs, and behaviors, with important implications for culture and welfare \citep{bisin2011economics, bisin2023advances}. Information may be lost in this transmission process, not only because of physical communication frictions, but also because of deliberate omissions, as different generations may have misaligned interests  \citep{hirshleifer2020presidential, stubbersfield2022content}. For example, parents who wish to encourage effort may emphasize narratives suggesting that hard work  leads to a better life \citep{benabou2006belief, doepke2017parenting}.

This paper analyzes the intergenerational transmission of information using an overlapping generations model that incorporates both exogenous communication frictions and misaligned interests across generations. As in \citet{dye1985disclosure}, we restrict attention to the transmission of \textit{hard information}, so that information can be concealed but cannot be fabricated or falsified. We show that 
as communication frictions vanish,
agents become increasingly selective in disclosing information, and in the limit, all signals except for the most favorable ones will be concealed. 

We study steady state equilibria of a doubly infinite time horizon game. The state of the world is constant over time and is either \textit{high} or \textit{low}. In each period, only one agent is active. He observes a private signal about the state and, with some probability, receives the signals his immediate predecessor passed on to him. The agent then chooses both an effort level and a subset of the signals he possesses (including his private signal and any inherited signals) to pass on to his immediate successor. With some probability, communication succeeds and the successor observes this subset. However, agents cannot directly observe others' actions, whether and when communication failed, or the exact sequence of the disclosed signals.

Each agent's optimal effort strictly increases in his belief about the state and this effort generates a positive externality for his predecessor. Therefore, when choosing the subset of signals to disclose, each agent's objective is to maximize his successor's belief about the state.

Due to the plethora of equilibria, some of which are driven by unreasonable off-path beliefs,\footnote{For example, when there are three or more signal realizations, there may exist equilibria where signal realizations with the highest likelihood ratio are concealed while those with the second-highest likelihood ratio are revealed. Such equilibria can be sustained by off-path beliefs in which an agent who observes the disclosure of any signal with the highest likelihood ratio infers that many signals with low likelihood ratios are concealed.} we restrict attention to equilibria that satisfy at least one of the following refinements. The first refinement requires that agents’ beliefs be \textit{monotone}, in the sense that their belief about the state increases when their predecessor substitutes a lower‑likelihood‑ratio signal for a higher‑likelihood‑ratio one. The second refinement requires a \textit{threshold equilibrium}, whereby agents disclose a signal if and only if its likelihood ratio exceeds a threshold that may depend on their private and inherited signals. Both refinements are satisfied by all equilibria of \citet{dye1985disclosure}'s model.

Theorem \ref{Theorem1} shows that when communication fails with low probability, disclosure becomes extremely selective. In particular, in all strict equilibria that satisfy at least one of our two refinements, agents conceal all signals except the ones with the highest likelihood ratio.

To better understand how selectively agents disclose as communication frictions vary, Theorem \ref{Theorem2} allows for arbitrary communication success rates. It focuses on \textit{constant-threshold equilibria}, that is, threshold equilibria with a disclosure threshold that is constant across all information sets. The theorem shows that, for every signal with a likelihood ratio greater than one, there exists a constant-threshold equilibrium where the signal is disclosed if and only if the communication success rate is below a cutoff. This cutoff is strictly increasing in the signal’s likelihood ratio. Consequently, as communication frictions vanish, disclosure becomes more selective, in the sense that a smaller set of signals can be revealed in constant-threshold equilibria.

Our findings contrast with \citet{dye1985disclosure}, who shows that a higher probability of the sender possessing evidence leads to less selective disclosure. This difference arises from the different forms of uncertainty agents face about others’ information structures. To build intuition, consider the sequence of realized signals following the last communication failure in our model. This sequence consists of multiple good signals (those that are disclosed), interspersed with (sub)sequences of bad signals (which may be of zero length). That is, there is one bad sequence preceding the first good signal, one between each pair of adjacent good signals, and one following the last good signal. Conditional on the state, the expected length of each bad sequence is independent of the  disclosed good signals. Hence, by disclosing another good signal, the agent also reveals the existence of another bad sequence. 
This negative inference is absent in \citet{dye1985disclosure}, where any disclosure implies that no information is hidden. As communication frictions vanish, the expected length of each bad sequence increases while the likelihood ratio of each good signal remains unchanged. Disclosure therefore becomes more selective, since each disclosed signal must have a higher likelihood ratio in order to overcome the stronger negative effect of one more bad sequence. 

Our work contributes to the literature on disclosure games where the sender has an unknown amount of signals.\footnote{\citet*{di2021strategic} study an agent who observes the realizations of the $k$ $(\leq n)$ highest draws among $n$ conditionally i.i.d. signals, which differs from our model where the number of signals disclosed is endogenous. In \citet{bardhi2023local}, each sender decides whether to disclose his signal without observing its realization.
In \citet{antic2025selected}, the state is $n$-dimensional and the sender selects which of $k$ $(\leq n)$ realized values of those dimensions to disclose (typically not the highest ones). 
\citet{onuchic2025disclosure} consider multiple senders who decide whether to disclose a multi-dimensional state via a collective decision rule.} Most of the existing works on this topic study one-shot games. The models of 
\citet{shin1994news, song2003disclosures} and
\citet{dziuda2011strategic} focus on binary signals, which are not well-suited to study the selectiveness of disclosure. \citet{gao2025inference} assumes that the sender has a continuum of signals and analyzes equilibria satisfying the \textit{truth-leaning refinement} of \citet*{hart2017evidence}, which, as we discuss in Section \ref{sec3}, conflicts with our refinements.\footnote{\citet{gieczewski2024coalition} propose a coalitional proofness refinement and provide conditions under which their refinement coincides with the truth-leaning refinement and the receiver-optimality refinement.} She shows that the sender sometimes discloses bad signals in order to replicate the signal distribution under another state, which contrasts with our findings.\footnote{\citet{gao2025inference} also shows that there exists a receiver-optimal equilibrium where the sender discloses only the most favorable signal, although there can also be other receiver-optimal equilibria.} 
\citet{rappoport2025evidence} focuses on receiver-optimal equilibria and 
shows that the receiver will take worse actions for the sender when the sender is expected to have more information. We instead examine the selectiveness of the sender’s disclosure behavior.

Unlike in one-shot disclosure games, the signals available to the senders in our overlapping generations  game are endogenously determined in equilibrium. This leads to qualitatively different findings compared to one-shot games, as we describe in Section \ref{sec5}.

A growing literature studies disclosure in other dynamic settings.
In \citet*{guttman2014not}, a sender may receive up to two signals over two periods and chooses both what to disclose and when. In \citet{felgenhauer2014strategic}, \citet{felgenhauer2017bayesian}, \cite{lou2023private}, \citet{arieli2025bayesian}, and \citet*{dai2026bayes}, a sender conducts a sequence of experiments before selecting a subset of outcomes to disclose to a receiver. In \citet{gieczewski2022verifiable} and \citet{squintani2025strategic}, 
some agents know the state and can disclose it to their neighbors.
In \citet*{kremer2024disclosing}, a potentially informed sender decides whether to disclose the current value of a random walk.
\citet{pei2025community} studies repeated games where players can voluntarily disclose signals about their past actions. 
Our paper complements this literature by examining how the selectiveness of disclosure depends on communication frictions.

Our work also contributes to the study of overlapping generations models starting from \citet{samuelson1958exact}.
While many papers in this literature \citep{bhaskar1998informational, acemoglu2015history}
focus on the observability of agents' actions, we study settings where actions are not observed, so that agents can influence their successors' behaviors only by disclosing information. 
\citet*{anderlini2012communication} analyze an overlapping generations model 
with strategic communication and payoff externalities. Communication in their model is cheap talk rather than the disclosure of hard evidence. They also focus on a different question -- whether agents learn the state in the long run -- while we study the selectiveness of disclosure. 

Broadly speaking, our model is also related to the literature on social learning such as \citet{banerjee1992simple}, \citet*{bikhchandani1992theory}, and \citet{smith2000pathological}.\footnote{\citet{niehaus2011filtered} studies the sharing of complementary or substitutable skills among a sequence of agents, which contrasts with our model where 
agents share noisy signals about a payoff-relevant state of the world.} Unlike those models, where there are no payoff externalities and agents learn from others’ actions in addition to their private signals, our model incorporates payoff externalities and agents learn from signals explicitly disclosed by others.


\section{Baseline Model}\label{sec2}
Consider a game with a doubly infinite time horizon $t \in \mathbb{Z} \equiv \{...-1,0,1,...\}$. Agent $t$ is the only active player in period $t$. He observes a  \textit{private signal} $s^t \in S$ about the state $\theta \in \{\underline{\theta},\overline{\theta}\}$, 
and with probability $\delta \in (0,1)$, some signals disclosed to him by agent $t-1$. Then he chooses an action $a^t \in \mathbb{R}$ and a set of signals to reveal to agent $t+1$. 
We assume that $\theta$ is constant over time, 
the set of signal realizations $S \equiv \{s_1,s_2,...,s_l\}$ is finite with $l \geq 2$, and that no signal reveals any state. Let $\pi_0 \in (0,1)$ denote the prior probability that $\theta=\overline{\theta}$.
Let $f_{\theta,i}$ denote the probability of $s_i$ conditional on state $\theta$.
To simplify notation, we write 
$\overline{f}_i$ instead of
$f_{\overline{\theta},i}$  and 
 $\underline{f}_i$ instead of
$f_{\underline{\theta},i}$. 
Let $\mathbf{f} \equiv \{\overline{f_i},\underline{f_i}\}_{i=1}^l$.
Let $L_i \equiv \overline{f_i}/\underline{f_i}$ denote the \textit{likelihood ratio} of signal $s_i$. We assume that $L_i \neq L_j$ for every $i \neq j$ and we index signals so that $L_i$ is strictly decreasing in $i$
\begin{equation}\label{eq2.1}
\infty > L_1 > L_2 >... > L_l > 0.\footnote{Analogs of our results can be derived when multiple signal realizations share the same likelihood ratio.} 
\end{equation}

Agent $t$'s payoff is $u(\theta,a^t)+ v(\theta,a^{t+1})$, which depends on the state $\theta$, his action $a^t$, and his immediate successor's action $a^{t+1}$. We assume that (i) $v$ is strictly increasing in $a^{t+1}$, (ii) $u$ is strictly concave in $a^t$, which implies that 
there is a unique $a^*(\pi) \in \mathbb{R}$ for every $\pi \in [0,1]$ that solves 
$\max_{a \in \mathbb{R}} \big\{ \pi u(\overline{\theta},a) + (1-\pi) u(\underline{\theta},a) \big\}$, 
and that (iii) $a^*(\pi)$ is strictly increasing in $\pi$.

The information agent $t-1$ reveals to agent $t$ is characterized by $\mathbf{h} \equiv (h_1,h_2,...h_l) \in H \equiv \mathbb{N}^l$ where $h_i$ stands for the number of signal $s_i$ revealed. With probability $\delta \in (0,1)$, communication succeeds, in which case agent $t$ observes the number of each signal revealed by agent $t-1$ but not the exact order of these signals, i.e., agent $t$'s \textit{history} is $\mathbf{h}$.
With probability $1-\delta$, communication fails in which case
agent $t$'s history is $\mathbf{0} \equiv (0,...,0)$.  We interpret $1-\delta$ as a communication friction. We refer to 
 $h_1+...+h_l$ as the \textit{length} of history $(h_1,...,h_l)$.
Each agent only observes his history and private signal, but cannot directly observe his predecessors' actions, the calendar times at which communication fails, the exact sequence of disclosed signals, and so on.


We assume that signals are verifiable so that they can be concealed but not falsified
\citep{grossman1981informational, milgrom1981good, dye1985disclosure}. 
Formally, let $\mathbf{1}_i \in \mathbb{N}^l$ denote a vector where the $i$th entry is $1$ and all other entries are $0$. After observing history $\mathbf{h}$ and receiving private signal $s_i$, 
agent $t$ can only disclose vectors that are no more than $\mathbf{h} + \mathbf{1}_i$. That is, he can only reveal a subset of the signals he possesses, which include his private signal and the signals he inherited (if any).

A \textit{pure strategy} $\sigma: H \times S \rightarrow \mathbb{R} \times H$ maps an agent's history and private signal to his action and the information he discloses to his immediate successor, subject to the feasibility constraint.
Our solution concept is \textit{steady state Bayesian equilibrium} (or \textit{equilibrium} for short), which extends the notion of steady state equilibrium in \citet*{clark2021record} and \citet{pei2025community} to incomplete information games. 
A (pure-strategy) equilibrium of our game consists of (i) a pure strategy $\sigma$, (ii) two distributions $\underline{\mu},\overline{\mu} \in \Delta (H)$ over histories, and (iii) a belief map $\pi: H \rightarrow [0,1]$ which represents the probability that an agent assigns to state $\overline{\theta}$ after observing his history but before observing his own private signal. These three components satisfy (i) when all agents use strategy $\sigma$, there exists a steady state in which agents' histories are distributed according to $\underline{\mu}$ conditional on $\theta=\underline{\theta}$ and are distributed according to $\overline{\mu}$ conditional on $\theta=\overline{\theta}$, 
(ii) $\pi(\mathbf{h})$ is derived via Bayes rule from $(\underline{\mu},\overline{\mu})$
 at every $\mathbf{h}$ that occurs with positive probability,
and that (iii) $\sigma$ maximizes an agent's expected payoff when all agents' beliefs after observing their histories are given by the mapping $\pi$.\footnote{As in \citet*{clark2021record}, our solution concept requires all agents using the same pure strategy, i.e., the same mapping $\sigma: H \times S \rightarrow \mathbb{R} \times H$. Because the distribution over histories is at a steady state in our equilibrium, whether an agent observes calendar time does not affect his belief about the state.}
A \textit{strict equilibrium} is one in which the equilibrium strategy $\sigma$ is strictly optimal for the agents at each history-signal pair $(\mathbf{h},s) \in H \times S$. 

The existence of strict equilibrium will be implied by  Proposition \ref{Prop1}. Lemma \ref{L0} shows that each pure strategy induces a unique steady state history distribution conditional on each state.
\begin{Lemma}\label{L0}
For every pure strategy $\sigma$ and conditional on every $\theta \in \{\overline{\theta},\underline{\theta}\}$, there exists a unique steady state distribution over histories $\mu_{\theta} \in \Delta (H)$.
\end{Lemma}
The proof is in Appendix \ref{secE}. 
\paragraph{Remarks:} One interpretation of our model is paternalistic parenting in the spirit of \citet{doepke2017parenting}. 
Children decide how to allocate time between schoolwork and tempting but unproductive activities (e.g., playing video games). The state  $\theta$ captures the return to studying in terms of future well-being and is either high ($\theta=\overline{\theta}$) or low ($\theta=\underline{\theta}$). Children devote more time to studying when they believe the return is higher. Parents, motivated by paternalistic concern, prefer their children to allocate more time to studying than to playing video games. 
Children cannot observe whether their parents themselves studied hard when they were young. However, 
parents can shape their children’s beliefs by sharing anecdotes—some based on their own observations (i.e., their private signals) and others passed down from their own parents.\footnote{This is referred to as \textit{authoritative parenting} in \citet{doepke2017parenting} that ``\textit{parents can mold their children's preferences so as to align them with their own, for instance, by emphasizing the virtue of hard work...}''} 

Our model can also be applied to an organizational setting. New members (e.g., workers) choose their effort levels subject to a convex cost, while senior members (e.g., managers) tell stories regarding the productivity of effort. Because workers are more inclined to work hard when they perceive their effort as more productive, managers—who do not internalize the workers' effort costs—strategically share stories to induce higher effort from the workers. Some of these stories are personal observations while others are 
inherited from previous generations. 
Workers cannot observe their managers' effort in previous periods. 
Workers age into managerial roles after one period and choose a set of stories to pass on to the next generation of workers. 

Our overlapping generations game is related to, albeit qualitatively different from, a one-shot disclosure game where the number of verifiable signals the sender observes is drawn from a geometric distribution with parameter $\delta$. We comment on the qualitative differences between the equilibria of our game and those in the one-shot disclosure game in Section \ref{sec5}.

In Sections \ref{secDiscConstraintN} and
\ref{sechistorydependent}, we discuss two extensions of our baseline model. This includes when agents face an upper bound on the number of signals disclosed, or when the communication success rate depends on the set of signals disclosed.


\section{Main Results}\label{sec3}
Our game admits many equilibria because each agent may possess arbitrarily many signals, allowing their successor's belief to become arbitrarily pessimistic at off-path histories. This in turn discourages agents from revealing those off-path histories, making those beliefs self-fulfilling. We start with some preliminary observations which describe some of those equilibria:
\begin{Proposition}\label{Prop1}
Fix $(u,v,\mathbf{f},\pi_0)$. For every $\delta \in (0,1)$, 
$i \in \{1,2,...,l\}$ with $L_i >1$, and $n \in \mathbb{N} \cup \{\infty\}$, 
there is a strict equilibrium where (i) agents never disclose any $s \neq s_i$ and (ii) the number of $s_i$ disclosed 
at history $\mathbf{h} = (h_1,...,h_l)$ and private signal $s_j$ is $\min \{n, h_i + \mathds{1}\{i=j\} \}$.
\end{Proposition}
The proof is in Appendix \ref{secC}. Proposition \ref{Prop1} implies that regardless of $\delta$, there is always an equilibrium where agents disclose all realized $s_1$ and nothing else. 
However, it also implies that for any $s_i$ with a likelihood ratio more than $1$, there is also a strict equilibrium in which agents never disclose anything other than $s_i$ and disclose up to $n$ realizations of $s_i$. Such equilibria exist even if there are other signals with strictly higher likelihood ratios than $s_i$, as they can be sustained by the off-path belief that as long as any signal other than $s_i$ (including those with higher likelihood ratios) is disclosed, the agent who disclosed it must have concealed many signals
with low likelihood ratios, which justifies their successors'  pessimistic beliefs about $\theta$. 

Such beliefs seem unreasonable since agents should become more \textit{optimistic} about the state once 
they inherit one more high-likelihood-ratio signal and one fewer low-likelihood-ratio signal.
We are unaware of standard refinements in signaling games that can rule out those unreasonable equilibria. Hence, we propose two refinements, each of which is sufficient to rule them out.\footnote{Analogs of our refinements are also used in \citet*{guttman2014not}. By the end of this section, we will discuss the relationship between our refinements and the \textit{truth-leaning refinement} in \citet*{hart2017evidence} as well as the receiver-optimal equilibria in \citet{rappoport2025evidence} and \citet{gao2025inference}.}

Our first refinement is on the belief map $\pi : H \rightarrow [0,1]$, which stands for the probability agents assign to state $\overline{\theta}$ after observing their history but before observing their private signal.  
\begin{Definition}
A belief map $\pi: H \rightarrow [0,1]$ is \textit{monotone} if $\pi(\mathbf{h}) > \pi(\mathbf{h^*})$
for every pair of histories $\mathbf{h} \equiv (h_1,...,h_l)$ and $\mathbf{h^*} \equiv (h_1^*,...,h_l^*)$ that have the same length and satisfy
\begin{equation}\label{3.1} 
\sum_{j=1}^n h_j \geq \sum_{j=1}^n h_j^* \textrm{ for every } n < l \textrm{ with strict inequality for some}.
\end{equation}
\end{Definition}
An equilibrium is a \textit{monotone equilibrium} if it has a monotone belief map. 
Intuitively, a belief map is monotone if fixing the total number of signals disclosed by agent $t$,  agent $t+1$ will become more optimistic about the state once agent $t$ replaces a signal that has a lower likelihood ratio with one that has a higher likelihood ratio.
Our monotonicity property is satisfied in all equilibria of \citet{dye1985disclosure}'s model where the sender 
either has one signal or has no signal.

Our second refinement is on agents' equilibrium disclosure behaviors. Let 
$\sigma_D (\mathbf{h},s) \in H \equiv \mathbb{N}^l$ denote the set of signals disclosed to the next agent upon observing history $\mathbf{h}$ and private signal $s$, which is a component of agents' strategy $\sigma$. 
For every $\mathbf{h} \in H$ and $j \in \{0,1,...,l\}$, let $\mathbf{h}[j] \in H$ denote an $l$-dimensional vector where the first $j$ entries coincide with $\mathbf{h}$ and the other entries are all $0$.  According to our definition, $\mathbf{h}[0]= \mathbf{0}$ and $\mathbf{h}[l]=\mathbf{h}$ for every $\mathbf{h} \in H$.  
\begin{Definition1}
A strategy $\sigma$ is a \textit{threshold strategy} if there exists a function $J: H \times S \rightarrow \{0,1,...,l\}$ such that 
$\sigma_D (\mathbf{h},s_i)= (\mathbf{h} +\mathbf{1}_i) [ J(\mathbf{h},s_i) ]$ for every $\mathbf{h} \in H$ and $i \in \{1,2,...,l\}$. 
\end{Definition1}
An equilibrium is a \textit{threshold equilibrium} if agents use a threshold strategy.\footnote{In any threshold equilibrium, agents will use a threshold strategy in equilibrium but need to have incentives not to deviate to any other strategy, including strategies that are not threshold strategies.} 
We view our threshold strategy as a natural extension of the equilibrium strategies in \citet{dye1985disclosure} to our dynamic setting. 
Intuitively, a threshold strategy is one where at every history-signal pair, agents reveal all available signals (including his private signal and the signals contained in his history) with likelihood ratios above some cutoff and conceal all signals with likelihood ratio below the cutoff.
It allows the disclosure threshold to vary across histories and private signals. 
For example, suppose $J(\mathbf{0},s)=1$ for every $s \in S$ and $J(\mathbf{h},s)=0$ for every $\mathbf{h} \neq \mathbf{0}$ and $s \in S$, then agents disclose $s_1$ when the length of his history is $0$ and disclose nothing otherwise. 
It rules out, for example, strategies where at some history-signal pair, agents disclose $s_2$ but conceal $s_1$.

Among the equilibria in Proposition \ref{Prop1}, the one where agents reveal all realized $s_1$ but nothing else satisfies both of our refinements. The ones where agents disclose $s_i$ for some $i \geq 2$ but conceal $s_1$ fail both refinements. Some natural questions include (i) whether there are threshold equilibria where agents reveal $\{s_1,s_2,...,s_i\}$ for some $i \geq 2$ and conceal other signals, (ii) whether there are equilibria with monotone belief maps where agents reveal signals other than $s_1$, and (iii) how does the sender's disclosure behavior depends on communication frictions.

Our main result, Theorem \ref{Theorem1}, shows that when communication frictions vanish, signals other than $s_1$ are never disclosed in any strict equilibrium that satisfies at least one of our refinements.
\begin{Theorem}\label{Theorem1}
Fix any $(u,v,\mathbf{f},\pi_0)$. There exists $\delta^* \in (0,1)$ such that for every $\delta > \delta^*$:
\begin{enumerate}
    \item In any strict monotone equilibrium, signals $s_2,...,s_l$ are never disclosed at any history that occurs with positive probability. 
        \item In any strict threshold equilibrium, signals $s_2,...,s_l$ are never disclosed at any history that occurs with positive probability.  
\end{enumerate}
\end{Theorem}
Theorem \ref{Theorem1} implies that as the communication friction vanishes, agents become extremely selective in disclosing information, in the sense that they will conceal all signals except for $s_1$, the signal that has the highest likelihood ratio.
The required magnitude of $\delta^*$ depends on the signal structure $\mathbf{f}$. For example, when $L_1-L_2$ is close to $0$, the value $\delta^*$ is close to $1$.

Together with Proposition \ref{Prop1}, Theorem \ref{Theorem1} leads to a full characterization of strict equilibria that satisfies either one of our refinements when $\delta$ is close to $1$:
\begin{Corollary}\label{cor1}
Fix any $(u,v,\mathbf{f},\pi_0)$. There exists $\delta^* \in (0,1)$ such that for every $\delta > \delta^*$:
\begin{enumerate}
\item Each strict monotone equilibrium is characterized by some $n \in \mathbb{N} \cup \{\infty\}$ such that $\sigma_D (\mathbf{h},s_i)= \min \{ n \mathbf{1}_1 , (\mathbf{h} +\mathbf{1}_i) [1] \}$ for every $\mathbf{h} \in H$ and $i \in \{1,2,...,l\}$. 
\item For every $n \in \mathbb{N} \cup \{\infty\}$, 
there exists a strict monotone equilibrium where 
 agents conceal all signals other than $s_1$ and disclose up to $n$ realizations of $s_1$.
\item There are only two strict threshold equilibria: one where agents never disclose anything and one where agents only disclose signal $s_1$ and disclose all realized $s_1$. 
\end{enumerate}
\end{Corollary}
The proofs of Theorem \ref{Theorem1} and Corollary \ref{cor1} are both in Section \ref{sub3.2}.

In order to better visualize how communication frictions affect the selectiveness of agents' disclosure,  
our next result allows for any value of $\delta$ and provides a full characterization of equilibria in which agents use 
\textit{constant-threshold strategies} (i.e., \textit{constant-threshold equilibria}). These are strategies where the disclosure threshold is independent of histories and private signals.
\begin{Definition2}
A strategy $\sigma$ is a \textit{constant-threshold strategy} if there exists $j \in \{0,1,...,l\}$
such that 
$\sigma_D (\mathbf{h},s_i)= (\mathbf{h} +\mathbf{1}_i) [j]$ for every $\mathbf{h} \in H$ and $i \in \{1,2,...,l\}$. 
\end{Definition2}
For example, agents will never disclose any signal when the threshold is constantly $0$, and will
disclose all signals when the threshold is constantly $l$. 
Theorem \ref{Theorem2} 
focuses on constant-threshold 
equilibria  and
shows that as $\delta$ increases, agents become increasingly selective in disclosing information, 
in the sense that a smaller set of signals can be disclosed in this class of equilibria.
\begin{Theorem}\label{Theorem2}
Fix any $(u,v,\mathbf{f},\pi_0)$ and $n \in \{2,...,l\}$.
\begin{enumerate}
\item If $L_n >1$, then there exists $\delta_n \in (0,1)$ such that there is a constant-threshold equilibrium where agents use a constant threshold $n$ \textit{if and only if} $\delta \leq \delta_n$. Furthermore, $\delta_n$ is strictly decreasing in $n$. 
\item If $L_n \leq 1$, then for any $\delta$, agents never disclose $s_n$ in any constant-threshold equilibrium.  
\end{enumerate}
\end{Theorem}
Theorem \ref{Theorem2} implies that (i) signals with likelihood ratio no more than $1$ will never be disclosed in any 
constant-threshold equilibrium and (ii) a signal with likelihood ratio strictly more than $1$ will be disclosed in some constant-threshold equilibria if and only if the communication success rate $\delta$ is below some cutoff. 
Furthermore, $\delta_2>\delta_3>...$, which means that
the cutoff $\delta$ is higher for signals with higher likelihood ratios, so higher likelihood ratio signals are disclosed in a larger range of $\delta$. 
Since any constant-threshold equilibrium is also a threshold equilibrium,
Theorem \ref{Theorem2} also implies that when $\delta > \delta_2$, there are only two constant-threshold equilibria: the one where agents never disclose anything and the one where agents 
disclose a signal if and only if it is $s_1$.

Our finding  that disclosure becomes more selective as $\delta$ increases contrasts with the results in \citet{dye1985disclosure}. In \citet{dye1985disclosure}'s model, as the probability that the sender has evidence increases, he becomes less selective in the sense that a larger set of realized signals are disclosed in equilibrium.  
In the limiting case studied by \citet{grossman1981informational} and \citet{milgrom1981good} where the sender has evidence for sure, no disclosure leads to the receiver to infer that the sender has the worst possible signal, which motivates the sender to disclose even the second worst signal.

In our model, a larger $\delta$ increases the chances that agents observe the signals disclosed by their predecessors, so absent strategic concerns, their successors should expect them to have more signals. However, disclosure becomes more selective due to the nature of the uncertainty their successors face about their information structures, which differs from that in \citet{dye1985disclosure}. 


To understand the differences, suppose agents are expected to disclose signals in $\{s_1,...,s_i\}$ (i.e., \textit{good signals}) and to conceal signals in $\{s_{i+1},...,s_l\}$ (i.e., \textit{bad signals}). After the last communication failure, the sequence of realized signals consists of some good signals and some subsequences of bad signals (i.e., \textit{bad subsequences}), where the length of each bad subsequence can be anything from $0$ to $\infty$ and is
independent of the number of good signals in the sequence. One useful observation is that the number of bad subsequences equals the number of good signals plus $1$. An example of such a sequence of signals  is depicted in Figure 1.

\begin{figure}
\begin{center}
\begin{tikzpicture}[scale=0.5]
\draw[thick, <->] (-10,0)--(-2,0)node[above]{1}--
(0,0)node[above]{2}--
(6,0)node[above]{3}--
(20,0)node[above]{4}--
(26,0);
\draw[fill=red] (-8,0) circle [radius=0.2];
\draw[fill=red] (-6,0) circle [radius=0.2];
\draw[fill=red] (-4,0) circle [radius=0.2];
\draw[fill=green] (-2,0) circle [radius=0.2];
\draw[dashed, <-] (0,2)--(-1,0.3);
\draw[dashed, <-] (-2,2)--(-1,0.3);
\node at (-1,2.4) {a bad subsequence with length $0$};
\draw[fill=green] (0,0) circle [radius=0.2];
\draw[fill=red] (2,0) circle [radius=0.2];
\draw[fill=red] (4,0) circle [radius=0.2];
\draw[fill=green] (6,0) circle [radius=0.2];
\draw[fill=red] (8,0) circle [radius=0.2];
\draw[fill=red] (10,0) circle [radius=0.2];
\draw[fill=red] (12,0) circle [radius=0.2];
\draw[fill=red] (14,0) circle [radius=0.2];
\draw[fill=red] (16,0) circle [radius=0.2];
\draw[fill=red] (18,0) circle [radius=0.2];
\draw[fill=green] (20,0) circle [radius=0.2];
\draw[fill=red] (22,0) circle [radius=0.2];
\draw[fill=red] (24,0) circle [radius=0.2];
\end{tikzpicture}
\caption{A sequence of realized private signals after the most recent communication failure, with red circles representing bad signals and green circles representing good signals. In this example, there are $4$ good signals and $5$ bad subsequences. The bad subsequence before good signal $1$ has length $3$, the bad sequence between good signals $1$ and $2$ has length $0$, the bad subsequence between good signals $2$ and $3$ has length $2$, the bad subsequence between good signals $3$ and $4$ has length $6$, and the bad subsequence after good signal $4$ has length $2$.}
\end{center}
\end{figure}
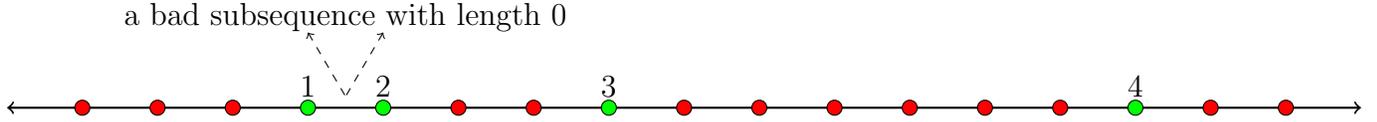

When agents observe the number of good signals but not the length of each bad subsequence, disclosing another good signal shows to the next agent that there is one more good signal and one more bad subsequence. The former has a positive effect on the next agent's belief as long as the likelihood ratio of the good signal exceeds $1$. The latter has a negative effect since bad subsequences are likely to be longer when the state is $\underline{\theta}$. This negative effect is absent in \citet{dye1985disclosure}'s model where disclosing any signal shows to the receiver that no information is hidden.

Therefore, the agents in our model have incentives to disclose a signal if and only if its likelihood ratio is large enough so that the direct effect dominates the indirect effect. As $\delta$ increases, the expected length of each bad subsequence increases conditional on each state, so the indirect effect of an additional bad subsequence 
becomes stronger yet the direct effect remains unchanged. Hence, agents become more selective in disclosing signals as fewer signals have high enough likelihood ratios to offset the increasingly pronounced indirect effect. 

Next,
we present a proof of Theorem \ref{Theorem2}, which helps readers to visualize mathematically the two effects of disclosing an additional signal, as well as how they depend on $\delta$. 

\begin{proof}[Proof of Theorem \ref{Theorem2}:]
 Since $v(\theta,a^{t+1})$ is strictly increasing in $a^{t+1}$ and agent $t+1$'s action is strictly increasing in the probability with which his belief assigns to $\overline{\theta}$, we know that at every $(\mathbf{h},s_i) \in H \times S$, agent $t$ will choose a vector $\mathbf{h'} \leq \mathbf{h} + \mathbf{1}_i$ that maximizes $\pi(\mathbf{h'})$. 

Suppose in equilibrium, all agents use a constant threshold strategy that discloses signals $s_1,...,s_n$. Under the steady state history distribution in state $\theta$,
the probability that an agent observes history $\mathbf{h} \equiv (h_1,h_2,...,h_n,0,...,0)$ is
\begin{equation}\label{3.2}
(1-\delta) \delta^{\overline{h}}  \cdot X_{\theta} (\mathbf{h}) \cdot \prod_{i=1}^n (f_{\theta,i})^{h_i},
\end{equation}
where $\overline{h} \equiv h_1+...+h_n$ denotes the length of history $\mathbf{h}$ and $X_{\theta} (\mathbf{h})$ denotes the coefficient in front of the term $\prod_{i=1}^n (p_i)^{h_i}$  in the following polynomial of $p_1,p_2,...,p_n$:
\begin{equation}\label{3.3}
\sum_{j=\overline{h}}^{\infty} \Big(
p_1+p_2+...+p_n + \delta (1-f_{\theta,1}-...-f_{\theta,n})
\Big)^j.
\end{equation}
A series of algebraic manipulations reveals that 
\begin{equation}\label{3.4}
 X_{\theta} (\mathbf{h})= \Big(1- \delta (1-f_{\theta,1}-...-f_{\theta,n}) \Big)^{-(\overline{h}+1)} \cdot \frac{\overline{h} !}{h_1 ! h_2 ! ... h_n ! }.  
\end{equation}
Therefore, 
\begin{equation}\label{3.5}
\frac{\pi(\mathbf{h})}{1-\pi(\mathbf{h})}=
\frac{\pi_0}{1-\pi_0} \cdot
\Big\{
\frac{1-\delta (1-\underline{f_1}-...-\underline{f_n})}{1-\delta (1-\overline{f_1}-...-\overline{f_n})}
\Big\}^{\overline{h}+1} \cdot 
\prod_{i=1}^n (L_i)^{h_i}. 
\end{equation}
Equation (\ref{3.5}) clearly shows the two effects of disclosing an additional signal. As we mentioned earlier, after the most recent communication failure, when the total number of signals disclosed is $\overline{h}$, there are $\overline{h}+1$ bad subsequences (some of them may have length zero), which explains why the power of the second term on the right-hand-side is $\overline{h}+1$.
By disclosing an additional signal $s_i$, the direct effect multiplies the likelihood ratio by $L_i$, as shown by the third term on the right-hand-side, 
and the indirect effect multiplies the likelihood ratio by 
\begin{equation}\label{3.6}
\frac{1-\delta (1-\underline{f_1}-...-\underline{f_n})}{1-\delta (1-\overline{f_1}-...-\overline{f_n})},
\end{equation}
which is the contribution of an additional bad subsequence. 
Crucially, as long as $\delta >0$ and $0<n<l$, the indirect effect, captured by (\ref{3.6}), is always strictly less than $1$. This is because each bad subsequence only consists of signals with likelihood ratios below some cutoff. The value of (\ref{3.6}) decreases in $\delta$ since a larger $\delta$ increases the expected length of each bad subsequence conditional on the state.  
Therefore, the value of (\ref{3.6}) is bounded below by its limit where $\delta \rightarrow 1$, which equals
\begin{equation}\label{3.7}
\frac{\underline{f_1}+...+\underline{f_n}}{\overline{f_1}+...+\overline{f_n}}.
\end{equation}
This lower bound is more than the inverse of $L_1$ but is less than the inverse of $L_n$. Hence, disclosing signal $s_1$ (the one with the highest likelihood ratio) is sufficient to offset the indirect effect. However, as long as signals other than $s_1$ are disclosed, disclosing the worst one among those signals is not sufficient to overcome the negative indirect effect when $\delta$ is large enough.

For a strategy with a constant threshold $n$ to be part of an equilibrium, a necessary condition is that 
agents prefer to disclose $(h_1,...,h_{n-1},h_n,0,...,0)$ to $(h_1,...,h_{n-1},h_n-1,0,...,0)$ for any $h_n \geq 1$ and $h_1,...,h_{n-1} \geq 0$. 
By (\ref{3.5}), this incentive constraint is satisfied if and only if 
\begin{equation}\label{3.8}
L_n \cdot \frac{1-\delta (1-\underline{f_1}-...-\underline{f_n})}{1-\delta (1-\overline{f_1}-...-\overline{f_n})} \geq 1. 
\end{equation}
This incentive constraint is also sufficient since $L_1>L_2>...>L_n$, so (\ref{3.8}) implies that 
\begin{equation*}
L_i \cdot \frac{1-\delta (1-\underline{f_1}-...-\underline{f_n})}{1-\delta (1-\overline{f_1}-...-\overline{f_n})} \geq 1 \textrm{ for every } i \leq n.
\end{equation*}
Since  
\begin{equation*}
\frac{1-\delta (1-\underline{f_1}-...-\underline{f_n})}{1-\delta (1-\overline{f_1}-...-\overline{f_n})} \leq 1
\end{equation*}
with strict inequality for every $0<n<l$, we know that $L_n>1$ is necessary for (\ref{3.8}). 
Since the left-hand-side of (\ref{3.8}) is strictly decreasing in $\delta$, which converges to $L_n$ as $\delta \rightarrow 0$ and converges to (\ref{3.7}) times $L_n$ as $\delta \rightarrow 1$,
we know that for every $n \geq 2$ that satisfies $L_n>1$, there exists $\delta_n \in (0,1)$ such that inequality (\ref{3.8}) holds if and only if $\delta \leq \delta_n$. Since $\delta_n$ is pinned down by the equation
\begin{equation*}
\underline{f_n}/\overline{f_n} = \frac{1-\delta_n (1-\underline{f_1}-...-\underline{f_n})}{1-\delta_n (1-\overline{f_1}-...-\overline{f_n})},
\end{equation*}
or equivalently,
\begin{equation}\label{deltan}
\delta_n^{-1} -1 = \frac{\underline{f_n}\sum_{j=1}^n \overline{f_j} - \overline{f_n} \sum_{j=1}^n \underline{f_j}}{\overline{f_n}-\underline{f_n}}.
\end{equation}
After some algebraic manipulations, and using the fact that
$L_n>1$, which implies that
$\overline{f_n}> \underline{f_n}$ 
and $\overline{f_{n-1}}> \underline{f_{n-1}}$,
we obtain that 
 $\delta_{n-1}> \delta_n$ if and only if 
\begin{equation*}
\Big\{ \sum_{j=1}^n (\overline{f_j}-\underline{f_j}) \Big\} \cdot \Big\{\underline{f_n} \overline{f_{n-1}}- \overline{f_n} \underline{f_{n-1}} \Big\} >0, 
\end{equation*}
which is true since $L_{n-1}>L_n$ and $\sum_{j=1}^n (\overline{f_j}-\underline{f_j})>0$ for any $n<l$.
\end{proof}

Theorem \ref{Theorem2} implies that once $\delta$ exceeds the cutoff $\delta_n$ defined in equation (\ref{deltan}), there is no constant-threshold equilibrium where agents disclose signal $s_n$.
This seems to suggest that once we focus on constant-threshold equilibria,
the expected amount of information an agent receives in the steady state
drops discontinuously at $\delta_n$.

Proposition \ref{thm:MIcont} shows that there is no such discontinuity. Focusing on the most informative constant-threshold equilibrium, i.e., the one with the highest disclosure threshold in terms of the signal index, the expected amount of information an agent receives, measured by the expected difference in 
Kullback–Leibler divergence between the agent's prior and his belief after observing his history (but before observing his private signal), is continuous and strictly increasing in $\delta$. 

Formally, let $\mathbf{h}_\delta^{(n)}\in H$  
denote the disclosed history observed by an agent in the steady state when all agents use constant threshold strategy $n$. 
For any $\delta \in (\delta_{n+1},\delta_{n}]$,
let $\mathbf{h}_\delta \equiv \mathbf{h}_\delta^{(n)}$ , so that $\mathbf{h}_{\delta}\in H$ is the disclosed history observed by an agent in the steady state when we select the equilibrium with the largest threshold $n$ that exists for each $\delta$. Theorem \ref{Theorem2} implies that
$\mathbf{h}_{\delta}$ is
well-defined. We denote the difference in the expected Kullback–Leibler divergence between an agent's prior belief about the state and his posterior belief by
\begin{equation}\label{eqMutualInfo}
I(\theta;\mathbf{h}_{\delta})=g(\pi_0)-\mathbb E\!\left[g\!\left(\Pr(\theta= \overline{\theta}\mid \mathbf{h}_{\delta})\right)\right],
\end{equation}
where $g(p) \equiv -p\log p-(1-p)\log(1-p)$ is the entropy function. We refer to $I(\theta;\mathbf{h}_{\delta})$ as \textit{mutual information}. Intuitively, $I(\theta;\mathbf{h}_{\delta})$ 
is greater when agents learn more about the state from their histories. 
Proposition \ref{thm:MIcont} shows that the mutual information is continuous and strictly increasing in $\delta$, even when the number of signals disclosed changes at the cutoff values of $\delta$. 
\begin{Proposition}\label{thm:MIcont}
The function
$I(\theta;\mathbf{h}_\delta)$ is continuous and strictly increasing in $\delta$.        
\end{Proposition}
The proof is in Appendix \ref{secF}. 

We conclude this section by discussing our refinements and comparing them with the ones in the existing literature. 
First, \citet*{guttman2014not} define \textit{monotone equilibria} and \textit{threshold equilibria} when
a sender can receive up to two signals in two periods.
We view our notions of monotone equilibria and threshold equilibria as natural extensions of their refinements to our setting where the sender can have any number of signals.

Second, Corollary \ref{cor1} implies that when $\delta$ is above some threshold $\delta^*$, for every strict threshold equilibrium, there exists a strict monotone equilibrium where agents use the same strategy but not vice versa. 
However, under general values of $\delta$,
it is unclear whether our monotone belief refinement is more permissive than our threshold strategy refinement.
Intuitively, agents' incentives to use a threshold strategy $J: H \times S \rightarrow H$ imply that they have no incentive to conceal any signal with an index below $J(\mathbf{h},s)$ at any $(\mathbf{h},s)$. Deviations where they conceal some of those signals will also reduce the length of the next agent's history (conditional on communication succeeds). The monotonicity refinement on the belief map has no bite in ruling out such deviations as it only imposes restrictions on histories with the same length. Section \ref{sec5} will revisit the differences between our two refinements when we discuss the differences between our overlapping generations game and one-shot disclosure games between a sender and a receiver where the sender has an unknown number of signals.

Third, as long as $L_2>1$ and
$\delta$ is large enough, both of our refinements conflict with the \textit{truth leaning refinement} in \citet*{hart2017evidence}. In our setting, their refinement requires that (A0) 
an agent will disclose all of his information if doing so maximizes his payoff
and (P0) agents believe that their predecessors did not hide any information upon observing any $\mathbf{h} \in H$ that occurs with zero probability in equilibrium. See page 700 of their paper for details.

We explain why there is no strict equilibrium that satisfies both our threshold strategy refinement and their (P0) requirement when $\delta$ is large. The conflict between our monotone belief refinement and their (P0) requirement follows from the same logic. From Corollary \ref{cor1}, we know that in every strict threshold equilibrium, agents will disclose $\mathbf{0}$ when their history is $\mathbf{0}$ and their private signal is $s_2$. According to Bayes rule, the probability with which their belief assigns to $\overline{\theta}$ upon observing history $\mathbf{0}$ is no more than their prior belief $\pi_0$. If such an equilibrium also satisfies (P0), then given that history $\mathbf{1}_2$ occurring with zero probability, $\pi(\mathbf{1}_2)$ must equal the agent's posterior belief when they know that their predecessor's history is $\mathbf{0}$ and their private signal is $s_2$. Since $L_2>1$, we have $\pi(\mathbf{1}_2) > \pi_0 > \pi (\mathbf{0})$. 
This contradicts our earlier conclusion that agents will disclose $\mathbf{0}$ upon receiving private signal $s_2$ at history $\mathbf{0}$.

Fourth, our refinements are neither stronger nor weaker relative to the receiver's optimal equilibrium used in \citet{rappoport2025evidence} and \citet{gao2025inference}. First, the equilibrium where the sender never discloses any signal satisfies both of our refinements. However, it is not the receiver-optimal equilibrium as the one where the sender only discloses $s_1$ leads to higher receiver welfare. Second, suppose 
$L_1>L_2>1>...$ but $\overline{f_1}$ and $\underline{f_1}$ are much smaller relative to $\overline{f_2}$ and $\underline{f_2}$. When $\delta$ is close to $1$, only signal $s_1$ can be disclosed in any strict equilibrium that satisfies any of our two refinements, whereas the receiver is strict better off in the equilibrium where only signal $s_2$ is disclosed. The latter equilibrium exists for any $\delta$ according to our Proposition \ref{Prop1}.

\section{Proof of Theorem 1}\label{sub3.2}
We state several useful properties of strict equilibria, without imposing the monotone belief refinement and the threshold strategy refinement. We start from a gradualism property.
\begin{Definition3}
For any $t \in \mathbb{Z}$, period $t$ history $\mathbf{h}^t \equiv (h_1^t,...,h_l^t)$, and 
period $t+1$
history $\mathbf{h}^{t+1} \equiv (h_1^{t+1},...,h_l^{t+1})$ that comes after $\mathbf{h}^t$, there exists at most one $i \in \{1,2,...,l\}$ such that $h_i^{t+1} > h_i^{t}$. If such an $i$ exists, then $h_i^{t+1}= h_i^t +1$.  
\end{Definition3}
Intuitively, this property is driven by our modeling assumption that each agent receives only one private signal in addition to the signals passed on to him by his predecessor.

Second, when an agent decides what information to disclose, his objective is to maximize his immediate successor's belief. Since he has the option to ignore his private signal and reveal his history $\mathbf{h}$, from which he will induce the same belief when communication succeeds, the history he discloses in equilibrium must lead to a weakly higher belief compared to $\pi(\mathbf{h})$. This leads to an \textit{improvement principle} similar to the one in \citet{banerjee2004word}.
\begin{Definition4}
For any equilibrium and any $t \in \mathbb{Z}$, if the realized history in period $t$ is $\mathbf{h}^t$ and agent $t$ chooses to disclose $\mathbf{h}^{t+1}$ after observing some private signal at $\mathbf{h}^t$, then $\pi(\mathbf{h}^{t+1}) \geq \pi(\mathbf{h}^t)$. If the equilibrium is strict and $\mathbf{h}^{t+1} \neq \mathbf{h}^t$, then 
$\pi(\mathbf{h}^{t+1}) > \pi(\mathbf{h}^t)$.\footnote{Our improvement principle is stated in terms of beliefs since for each agent in our model, maximizing his expected payoff from his successor's action is equivalent to choosing $\mathbf{h}$ in order to maximize $\pi(\mathbf{h})$. The equivalence between maximizing payoffs and maximizing beliefs breaks down in one of our extensions, in which case a more general version of the improvement principle holds: Disclosing $\mathbf{h}^{t+1}$ leads to a weakly greater expected payoff for agent $t$ compared to disclosing $\mathbf{h}^t$, where the inequality is strict when the equilibrium is strict and $\mathbf{h}^{t+1} \neq \mathbf{h}^t$.}
\end{Definition4}
An implication of our improvement principle is the following \textit{no-turning-back property}.
\begin{Definition5}
In any strict equilibrium, for any history $\mathbf{h}$ that occurs with positive probability in this equilibrium as well as any $\mathbf{h'} < \mathbf{h}$, history $\mathbf{h'}$ will not occur after history $\mathbf{h}$ before the next communication failure. 
\end{Definition5}
This is because agents can disclose any $\mathbf{h'} < \mathbf{h}$ as long as they can disclose $\mathbf{h}$, so their incentive to disclose $\mathbf{h}$ in a strict equilibrium implies that $\pi(\mathbf{h})> \pi(\mathbf{h'})$. Iteratively apply the improvement principle, we know that no one will disclose $\mathbf{h'}$ after $\mathbf{h}$ until communication fails.

Using these properties, one can see why Corollary \ref{cor1} is a direct implication of Theorem \ref{Theorem1}. Since only $s_1$ can be disclosed when $\delta$ is large enough and any strict equilibrium needs to respect the no-turning-back property, there are only two threshold strategies that meet both requirements: one where no signal is disclosed and one where agents disclose all realized $s_1$.\footnote{The no-turning-back property rules out strategies where at histories that occur with positive probability and contain a positive number of $s_1$, 
agents set the disclosure threshold to be $0$, i.e., conceal all realized signals.} 
When $\delta$ is large enough, strict monotone equilibrium turns out to be a more permissive concept: any strategy where agents disclose up to $n$ realizations of $s_1$ can be part of such an equilibrium as long as the belief map $\pi$ assigns a low value at any history with length more than $n$.

We now return to the proof of Theorem \ref{Theorem1}. Fix any strict equilibrium $(\sigma,\underline{\mu},\overline{\mu},\pi)$. For convenience, we write $\mu_{\theta}$ as the probability measure over histories induced by $\sigma$ conditional on  state $\theta$.  
Let $H_{\sigma} \subset H$ denote the set of histories that occur with positive probability under $\underline{\mu}$. Let $j$ denote the largest integer in $\{0,1,...,l\}$ such that there exists $\mathbf{h} \equiv (h_1,...,h_l)\in H_{\sigma}$ with $h_j \geq 1$. Our gradualism property implies that there exists $\mathbf{h} \in H_{\sigma}$ with $h_j=1$. 

Let $H(j) \subset H_{\sigma}$ denote the set of positive-probability histories under $\sigma$ such that the $j$th entry is $1$. Let $H (j-1) \subset H(j)$ denote the subset of $H(j)$ such that a history $\mathbf{h}$ belongs to $H(j-1)$ if and only if it minimizes $h_1+...+h_{j-1}$ among all vectors in $H(j)$. Iteratively define $H(j-2),...,H(1)$ with 
$H(1) \subset H(2) \subset ... \subset H(j-2) \subset H(j-1)$ 
such that for every $i \in \{1,2,...,j-2\}$, history
$\mathbf{h}$ belongs to $H(i)$ if and only if it minimizes $h_1+...+h_i$ among all vectors in $H(i+1)$. This leads to a unique vector in $H(1)$, which we denote by $\mathbf{h^*} \equiv (h_1^*,...,h_l^*)$. 
By construction, $h_{j+1}^*=...=h_l^*=0$ and $h_j^*=1$.

\paragraph{Strict Monotone Equilibria:}  Suppose for some $\delta$ is close enough to $1$, there exists a strict monotone equilibrium under which $j \geq 2$. The monotonicity of the belief map $\pi$ implies that $\pi(\mathbf{h^*}) \leq \pi(\mathbf{h})$ for any $\mathbf{h} \in H (j-1)$. We establish the following lemma, which shows that $\mathbf{h^*}$ can only originate from $\mathbf{h^*}-\mathbf{1}_j$.
\begin{Lemma}\label{L2}
Fix any strict monotone equilibrium. If an agent's history is $\mathbf{h^*}$, then his immediate predecessor's history is either $\mathbf{h^*}$ or $\mathbf{h^*}-\mathbf{1}_{j}$. 
\end{Lemma}
The proof is in Appendix \ref{secB}.

For every $\mathbf{h} \in H_{\sigma}$, let $D(\mathbf{h}) \equiv \{i:\sigma_D(\mathbf{h},s_i) \neq \mathbf{h}\} \subset \{1,2,...,l\}$. Given the no turning-back property, this implies that for any $\mathbf{h} \in H_{\sigma}$, $i \in D(\mathbf{h})$ if and only if the disclosed vector after receiving private signal $s_i$ at history $\mathbf{h}$ is not weakly less than $\mathbf{h}$, or equivalently, the number of signal $s_i$ disclosed is more than the $i$th entry of $\mathbf{h}$. 
For any $\mathbf{h}=(h_1,...,h_l)$ with $h_i \neq 0$ for some $i \geq 2$, we have $D(\mathbf{h}) \neq \emptyset$. This is because $\pi$ being monotone implies that $\pi(\mathbf{h}-\mathbf{1}_i+ \mathbf{1}_1) > \pi(\mathbf{h})$, so $1 \in D(\mathbf{h}) \neq \emptyset$. 
Whenever the set $D(\mathbf{h})$ is non-empty, let $\overline{d} (\mathbf{h}) \equiv \max D(\mathbf{h})$.

Since the equilibrium is strict, we know that 
 $\pi(\mathbf{h}) > \pi(\mathbf{h'})$
for every $\mathbf{h} \in H_{\sigma}$ and $\mathbf{h'}< \mathbf{h}$. Hence, if the agent receives private signal $s_i$ at history $\mathbf{h}$ with $i \notin D(\mathbf{h})$, he will disclose history $\mathbf{h}$ to his immediate successor, so the history will remain the same conditional on communication succeeds. This together with Lemma \ref{L2} leads to
\begin{equation}\label{3.11}
\mu_{\theta} (\mathbf{h^*}) = \mu_{\theta}(\mathbf{h^*}-\mathbf{1}_j) \cdot \delta f_{\theta,j} \cdot 
\sum_{n=0}^{\infty} \delta^n (1-\sum_{i \in D( \mathbf{h^*} )} f_{\theta,i})^n
= \mu_{\theta}(\mathbf{h^*}-\mathbf{1}_j) \cdot \frac{\delta f_{\theta,j} } {1-\delta (1-\sum_{i \in D(\mathbf{h^*})} f_{\theta,i})}.
\end{equation}
Equation (\ref{3.11}) implies that $\pi(\mathbf{h^*}) \geq \pi(\mathbf{h^*}-\mathbf{1}_j)$ if and only if 
\begin{equation}\label{3.12}
L_j \cdot \frac{1-\delta (1-\sum_{i \in D(\mathbf{h^*})} \underline{f_i})}{1-\delta (1-\sum_{i \in D(\mathbf{h^*})} \overline{f_i})} \geq 1.
\end{equation}
Recall our hypothesis that $j \geq 2$, which by our earlier conclusion, implies that $D(\mathbf{h^*})$ is non-empty and does not contain anything strictly greater than $j$. Since $S$ is a finite set, we know that there exists $\delta^* \in (0,1)$ such that for every $\delta>\delta^*$, inequality (\ref{3.12}) fails for any $j$ and non-empty $D(\mathbf{h^*})$ that satisfy $j \geq 2$ and $j \geq \max D(\mathbf{h^*})$. 
When 
$\delta>\delta^*$,
the above contradiction rules out strict monotone equilibria where 
signals other than $s_1$ are disclosed on the equilibrium path.

\paragraph{Strict Threshold Equilibria:} Suppose for some $\delta$ is close to $1$, there exists a strict threshold equilibrium where agents use a threshold strategy $\sigma$ under which $j \geq 2$. We establish the following lemma, which is the analog of Lemma \ref{L2} when the solution concept is threshold equilibrium. 
\begin{Lemma}\label{L1}
Fix any threshold equilibrium.
If an agent's history is $\mathbf{h^*}$, then his immediate predecessor's history is either $\mathbf{h^*}$ or $\mathbf{h^*}-\mathbf{1}_{j}$. 
\end{Lemma}
The proof is in Appendix \ref{secA}.

Next, consider the agent's disclosure behavior at history $\mathbf{h^*}$. Since the equilibrium is strict, $\pi(\mathbf{h^*}) > \pi(\mathbf{h})$ for any $\mathbf{h} < \mathbf{h^*}$. The definition of $j$ implies that upon receiving any signal $s_k$ with $k > j$ at $\mathbf{h^*}$, the disclosure threshold must be $j$, i.e., signals with index more than $j$ are all concealed and those with index no more than $j$ are all revealed. Upon receiving any signal $s_i$ with $i \leq j$ at $\mathbf{h^*}$, the history in the next period will not be $\mathbf{h^*}$ regardless of the threshold the agent uses. Hence, conditional on state $\theta$ and the current-period history being $\mathbf{h^*}$, the probability with which history remains at $\mathbf{h^*}$ in the next period is $\delta (1-\sum_{i=1}^j f_{\theta,i} )$.

Recall from Lemma \ref{L1} that when an agent's history is $\mathbf{h^*}$, his immediate predecessor's history is either $\mathbf{h^*}$ or $\mathbf{h^*}-\mathbf{1}_{j}$. The probability with which history $\mathbf{h^*}$ occurs conditional on state $\theta$ is
\begin{align}\label{3.9}
\mu_{\theta} (\mathbf{h^*}) 
&= \mu_{\theta} (\mathbf{h^*}-\mathbf{1}_j) \cdot \delta f_{\theta,j} \cdot 
\sum_{k=0}^{\infty} \delta^k \left(1-\sum_{i=1}^j f_{\theta,i}\right)^k \notag \\
&= \mu_{\theta} (\mathbf{h^*}-\mathbf{1}_j) \cdot 
\frac{\delta f_{\theta,j}}{1-\delta \left(1-\sum_{i=1}^j f_{\theta,i}\right)}.
\end{align}
After receiving private signal $s_j$ at history $\mathbf{h^*}-\mathbf{1}_j$, 
an agent has an incentive to disclose history $\mathbf{h^*}$ to his successor only if $\pi(\mathbf{h^*}) \geq \pi (\mathbf{h^*}-\mathbf{1}_j)$, which according to (\ref{3.9}), translates into
\begin{equation}\label{3.10}
L_j \cdot \frac{1-\delta (1-\sum_{i=1}^j \underline{f_i})}{1-\delta (1-\sum_{i=1}^j \overline{f_i})} \geq 1.
\end{equation}
When $\delta$ is above some cutoff, inequality (\ref{3.10}) fails for any $j \geq 2$, which leads to a contradiction.

\section{Discussion}
In Section \ref{sec5}, we discuss the connections between our overlapping generations model and one-shot disclosure games between a sender and a receiver where the number of signals the sender has follows a geometric distribution.
In Sections \ref{secDiscConstraintN} and
\ref{sechistorydependent}, we study two extensions of our baseline model, one where agents face a constraint on the maximal number of signals they can disclose and one where the communication success rate depends on the set of signals transmitted. 

\subsection{Comparison to one-shot disclosure games}\label{sec5}
A natural question is  whether our conclusions extend to a one-shot disclosure game between a sender and a receiver
where the sender has an unknown number of signals. 
Suppose the state is binary $\theta \in \{\overline{\theta},\underline{\theta}\}$.
The sender wants to maximize the probability the receiver's belief assigns to state $\overline{\theta}$.
The number of verifiable (conditionally independent) signals the sender has is drawn from a geometric distribution with parameter $\delta$. The receiver only observes the number of each signal the sender discloses, so that each receiver's information set belongs to the set $H \equiv \mathbb{N}^l$.

In this game, 
the sender's strategy is $\sigma_D: H \rightarrow H$ subject to the constraint that $\sigma_D(\mathbf{h}) \leq \mathbf{h}$ for every $\mathbf{h} \in H$. For any $\mathbf{h} \in H$ and $i \in \{0,1,...,l\}$, 
recall that  $\mathbf{h}[i] \in H$ is defined as a vector where the first $i$ entries coincide with $\mathbf{h}$ and the other entries are zero. We say that $\sigma_D$ is a \textit{threshold strategy} if there exists $J : H \rightarrow \{0,1,...,l\}$ such that $\sigma_D(\mathbf{h})=\mathbf{h}[J(\mathbf{h})]$ for every $\mathbf{h} \in H$. A threshold strategy is a \textit{constant threshold strategy} if $J (\mathbf{h})=J (\mathbf{h'})$ for every $\mathbf{h} \neq \mathbf{h'}$.  
The receiver's belief map is $\pi : H \rightarrow [0,1]$. 
The monotonicity property of the belief map is defined in the same way as in our overlapping generations model.  
The solution concept is Perfect Bayesian equilibrium. Our definitions of monotone equilibrium, threshold equilibrium, and constant threshold equilibrium extend naturally to this one-shot game.

Statement 1 of our Theorem \ref{Theorem1}, which concerns strict monotone equilibria, fails to extend to this one-shot game, thereby distinguishing the one-shot game from our overlapping-generations game.
Intuitively, this is because the gradualism property stated in Section \ref{sub3.2} breaks down, which means that Lemma \ref{L2} is no longer true. In the one-shot game, the sender will observe every combination of signals with positive probability, as opposed to only a subset of those combinations determined by players' equilibrium strategy in our overlapping generations game.

In what follows, we construct strict monotone equilibria of the one-shot game under which signals other than $s_1$ is disclosed with positive probability. 
We introduce a class of strategies for the sender parameterized by integer $k \in \mathbb{N}$.
\begin{itemize}
\item Strategy $k$: For every $j \neq k+1$, if the sender observes $j$ copies of signal $s_1$, then he discloses all the $s_1$ he has and nothing else. If the sender has $k+1$ signal $s_1$ but no $s_2$, then he discloses $k$ signal $s_1$ and nothing else. If the sender has $k+1$ signal $s_1$ and at least one $s_2$, then he discloses $k+1$ signal $s_1$, one $s_2$, and nothing else. 
\end{itemize}

Strategy $k$ is \textit{not} a threshold strategy. This is because when the sender has $k+1$ signal $s_1$ and two  $s_2$, he discloses only one $s_2$ and conceals the other. If the sender uses Strategy $k$ in the one-shot game, 
$s_2$ will be disclosed at some positive-probability information set since the sender will observe $k+1$ signal $s_1$ and one $s_2$ with positive probability.
However, if agents use Strategy $k$ in our overlapping generations game, 
$s_2$ will \textit{not} be disclosed on the equilibrium path as history $(k+1) \mathbf{1}_1$ will occur with zero probability.
\begin{Proposition}\label{Prop4}
For any 
$\delta \in (0,1)$, there exists $k \in \mathbb{N}$ such that there exists
a strict monotone equilibrium in the one-shot disclosure game
where the sender uses Strategy $k$. 
\end{Proposition}
The proof is in Appendix \ref{secG}.

In contrast,  Statement 2 of Theorem \ref{Theorem1} which concerns strict threshold equilibria, can be extended to this one-shot disclosure game.
\begin{Proposition}\label{Prop3}
There exists $\delta^* \in (0,1)$ such that when $\delta > \delta^*$, the sender will never disclose $s_2,...,s_l$ in any strict threshold equilibrium.
\end{Proposition}
The proof is in Appendix \ref{secD}.

Theorem \ref{Theorem2}, which focuses on constant-threshold equilibria, also extends to this one-shot disclosure game. This is because when the sender uses a constant disclosure threshold $n$, the only relevant incentive constraint is for him to disclose signal $s_n$ whenever it is available, i.e., $\pi(\mathbf{h}[n]) \leq \pi(\mathbf{h}[n]+\mathbf{1}_n)$ for every $\mathbf{h} \in H$. This is the case if and only if inequality (\ref{3.8}) holds. The rest of argument follows from the proof of Theorem \ref{Theorem2} in Section \ref{sec3}, which is omitted.

\subsection{Extension to disclosure constraints}\label{secDiscConstraintN}
We now discuss two extensions of our baseline model where our main result Theorem \ref{Theorem1} extends.
In our baseline model, agents can disclose any number of available signals. In practice, agents may face constraints on the number of signals they can disclose. Consider an extension where each agent can disclose at most $N \in \mathbb{N}$ signals. That is, after receiving private signal $s_i$ at history $\mathbf{h}$, he can disclose $\mathbf{h'}=(h_1',...,h_l')$ if and only if $\mathbf{h'} \leq \mathbf{h} + \mathbf{1}_i$ and $\sum_{j=1}^l h_j' \leq N$.

Our proof of Theorem \ref{Theorem1} extends to any value of $N$. To see this, notice that the gradualism property, the improvement principle, and the no-turning-back property remain true regardless of $N$. Lemma \ref{L2} and Lemma \ref{L1} also extend since their proofs only need to rule out the possibility that the current-period history is $\mathbf{h^*}$ yet the previous-period history is neither $\mathbf{h^*}$ nor $\mathbf{h^*}-\mathbf{1}_j$. Such a possibility is also ruled out when the agents face an additional constraint when they decide which signals to disclose. The calculations after Lemma \ref{L2} and Lemma \ref{L1} concern the probability of history $\mathbf{h^*}$, which by construction, must satisfy our additional constraint. Equations (\ref{3.11}) and (\ref{3.9}) express the probability of $\mathbf{h^*}$ as a function of the probability of history $\mathbf{h^*}-\mathbf{1}_j$. Despite the disclosure set $D (\mathbf{h^*})$ may depend on $N$, the way we derived a contradiction only uses a particular property of $D(\mathbf{h^*})$, which is $j \geq \max D(\mathbf{h^*})$, which holds regardless of $N$.

Nevertheless, Theorem \ref{Theorem1} characterizes a common property of all strict monotone equilibria and strict threshold equilibria when $\delta$ is large. The set of strict monotone equilibria and the set of strict threshold equilibria will depend on the upper bound $N$. For example, disclosing all realized $s_1$ will not be an equilibrium since the agents' strategy will violate our new constraint.

\subsection{Extension to message-dependent communication success rate}\label{sechistorydependent}
In our baseline model, the probability that communication succeeds is independent of the set of signals an agent chooses to disclose to their successor. In practice, however, the success rate may depend on the number of the signals disclosed. For example, children may absorb information from their parents with lower probability when parents attempt to convey too much.

Suppose now that when agent $t$ discloses $\mathbf{h}$, the probability that agent $t+1$ receives $\mathbf{h}$ is $\delta (\mathbf{h}) \equiv \exp (-\Delta \cdot p (\mathbf{h}))$ where $\Delta>0$ is a parameter
and $p: H \rightarrow [\underline{p},\overline{p}]$ for some $0<\underline{p}< \overline{p}< \infty$ such that $p(\mathbf{h})$ depends on 
$\mathbf{h}$ only through its length 
and $p(\mathbf{h})$ weakly increases as the length of $\mathbf{h}$ increases.
In this setting, communication is more likely to fail when the history disclosed is longer. 
As $\Delta \rightarrow 0$, communication success rate converges to $1$.

Theorem \ref{Theorem1} extends to this more general setting. To see this, notice that the gradualism property and the no-turning-back property hold independently of the communication success rate. A modified version of the improvement principle holds in this case, that when $\mathbf{h}^{t+1}$ is disclosed when the current-period history is $\mathbf{h}^t$,  the expected value of $v(\theta,a^*(\pi))$ when the agent discloses $\mathbf{h}^{t+1}$ is greater than the expected value when he discloses $\mathbf{h}^t$. We explain why Lemma \ref{L2} extends by the end of Appendix \ref{secB}. The same logic applies to the proof of Lemma \ref{L1}, which we omit in order to avoid repetition. 

In order to show that agents will never disclose any $s_2,...,s_l$ 
 in any strict monotone equilibrium when $\Delta$ is small enough, we need to modify 
equation (\ref{3.11}) as
\begin{equation}\label{7.1}
\mu_{\theta} (\mathbf{h^*}) = \mu_{\theta}(\mathbf{h^*}-\mathbf{1}_j) \cdot \delta (\mathbf{h^*}) \cdot f_{\theta,j} \cdot 
\sum_{n=0}^{\infty} \delta(\mathbf{h^*})^n \cdot (1-\sum_{i \in D( \mathbf{h^*} )} f_{\theta,i})^n.
\end{equation}
This is because when the current-period history is $\mathbf{h^*}-\mathbf{1}_j$, the history in the next period is $\mathbf{h^*}$ if and only if the current-period private signal is $s_j$, 
the agent discloses $\mathbf{h^*}$, and communication succeeds. Conditional on state $\theta$, the above event occurs with probability 
$\delta (\mathbf{h^*}) \cdot f_{\theta,j}$. Moreover, when the current history is $\mathbf{h^*}$, the next period's history remains at $\mathbf{h^*}$ if and only if (i) the agent's private signal does not belong to $D(\mathbf{h^*})$ in which case he discloses $\mathbf{h^*}$ to the next agent
and (ii) communication succeeds. Conditional on state $\theta$, the above event occurs with probability $\delta (\mathbf{h^*}) \cdot (1-\sum_{i \in D(\mathbf{h^*})} f_{\theta,i})$. Equation (\ref{7.1}) implies that 
$\pi(\mathbf{h^*}) < \pi(\mathbf{h^*}-\mathbf{1}_j)$ if and only if 
\begin{equation}\label{7.2}
L_j \cdot \frac{1-\delta(\mathbf{h^*}) \cdot  (1-\sum_{i \in D(\mathbf{h^*})} \underline{f_i})}{1-\delta(\mathbf{h^*}) \cdot (1-\sum_{i \in D(\mathbf{h^*})} \overline{f_i})} <1.
\end{equation}
Since $\delta (\mathbf{h^*})= \exp (-\Delta \cdot p (\mathbf{h^*}))$ and the function $p(\cdot)$ is uniformly bounded from above and from below, we know that there exists $\Delta^* >0$ such that when $\Delta < \Delta^*$, inequality (\ref{7.2}) holds for all 
$j$ and $D(\mathbf{h^*})$ that satisfy
$j \geq 2$ and $j \geq \max D(\mathbf{h^*})$. 
After receiving private signal $s_j$ at history $\mathbf{h^*}-\mathbf{1}_j$, 
(i) by disclosing $\mathbf{h^*}$, an agent induces public belief $\pi(\mathbf{h^*})$ with probability $\delta (\mathbf{h^*})$ and induces public belief $\pi(\mathbf{0})$ with probability $1-\delta (\mathbf{h^*})$; (ii)
by disclosing $\mathbf{h^*}-\mathbf{1}_j$, an agent induces public belief $\pi(\mathbf{h^*}-\mathbf{1}_j)$ with probability $\delta (\mathbf{h^*}-\mathbf{1}_j)$ and induces public belief $\pi(\mathbf{0})$ with probability $1-\delta (\mathbf{h^*}-\mathbf{1}_j)$. 
Disclosing $\mathbf{h^*}-\mathbf{1}_j$
leads to a strictly higher expected payoff than disclosing $\mathbf{h^*}$ 
when $\Delta$ is small enough since 
$\pi(\mathbf{h^*}) < \pi(\mathbf{h^*}-\mathbf{1}_j)$,
$\delta(\mathbf{h^*}) \leq \delta(\mathbf{h^*} -\mathbf{1}_j)$, and from the improvement principle
$\pi(\mathbf{0}) \leq \min \{ \pi(\mathbf{h^*}), \pi(\mathbf{h^*}-\mathbf{1}_j) \}$. This leads to a contradiction, which rules out strict monotone equilibria where agents disclose signals other than $s_1$.
The same reasoning can be applied to show that 
agents will never disclose $s_2,...,s_l$
in any strict threshold equilibrium when $\Delta$ is small enough, by replacing 
$\sum_{i \in D( \mathbf{h^*} )} f_{\theta,i}$ in (\ref{7.1}) and (\ref{7.2}) with $\sum_{i=1}^j f_{\theta,i}$.

\section{Conclusion}\label{sec7}
This paper studies the intergenerational transmission of verifiable information in an overlapping generations model that incorporates both physical communication frictions as well as agents' strategic incentives to conceal unfavorable information. We show that lowering the communication frictions will make the agents more selective in disclosing information. This is because disclosing an additional good signal has two effects on the next agent's belief: a direct effect which shows the existence of another good signal as well as an indirect effect which shows the existence of another subsequence of bad signals. As communication frictions vanish, the expected length of each bad subsequence increases conditional on each state. 
Hence, the indirect effect will become stronger, so that fewer signals can have high enough likelihood ratios to overcome the more pronounced indirect effect. The indirect effect is absent in \citet{dye1985disclosure}'s model since any disclosure proves to the receiver that no information is hidden. As a result, 
an increase in the probability that the sender has information makes his disclosure less selective.

We conclude by discussing several limitations of our analysis. First, as in the classic overlapping generations model of \citet{samuelson1958exact}, followed by \citet{bhaskar1998informational}, each agent's payoff depends only on his action and his immediate successor's action,\footnote{Allowing agents' payoffs to depend also on their predecessors' actions, as in \citet{acemoglu2015history}, will not change our results since even with such a dependence, each agent's objective is still to maximize their immediate successor's belief about the state when they choose the set of signals to disclose.} but not on future agents' actions. Suppose each agent $t$'s payoff depends not only on agent $t+1$'s action but also on the action of agent $t+2$'s. When agent $t$ decides which set of signals to disclose, maximizing agent $t+1$'s belief is no longer sufficient since the signals he disclosed also affects agent $t+1$'s ability to influence agent $t+2$'s belief, which in turn affects agent $t$'s payoff. Moreover, aside from the set of signals available to agent $t$, agent $t$'s 
posterior belief about $\theta$ also affects his incentives to disclose signals since it affects his predictions regarding the realized private signal of agent $t+1$, which together with the signals disclosed by agent $t$, affects agent $t+1$'s disclosure behavior. Solving this dynamic optimization problem is challenging, which we left for future research.

Second, we assume that agent $t$'s payoff takes the form of $u(\theta,a^t) + v(\theta,a^{t+1})$, that is, agent $t$'s marginal utility from his own action is independent of agent $t+1$'s action. If agent $t$'s payoff takes the general form $w(\theta,a^t,a^{t+1})$, then as in the previous case, agent $t$ also needs to consider agent $t+1$'s ability to influence agent $t+2$'s belief since agent $t+1$'s incentive to take actions depends not only on his belief about $\theta$ but also on his belief about agent $t+2$'s action.  The same challenge arises when solving this dynamic optimization problem. 

Third, as in \citet*{clark2021record}, our main result, Theorem \ref{Theorem1}, restricts attention to \textit{strict} steady state equilibria.\footnote{Our Theorem \ref{Theorem2}, which characterizes all constant-threshold equilibria, do not require the strictness of players' incentives, as indicated in the statement of that result.} We do not know whether Theorem \ref{Theorem1} holds when we weaken our solution concept
from strict monotone equilibria and strict threshold equilibria to pure-strategy monotone equilibria and pure-strategy threshold equilibria. One challenge is that no-turning-back property breaks down when we do not require players to have strict incentives, in which case Lemma \ref{L2} as well as the calculations after Lemma \ref{L2} and Lemma \ref{L1} no longer apply. For example, in threshold equilibria, once the history reaches $\mathbf{h^*}$ defined in Section \ref{sub3.2}, agents may have an incentive to conceal signal $s_j$ upon receiving signal $s_i$ with $i>j$ if 
$\pi(\mathbf{h^*}[j])= \pi( \mathbf{h^*}[k] )$ for some $k<j$. In equilibrium, this may in turn boost the next agent's belief upon observing history $\mathbf{h^*}$ as leaving history $\mathbf{h^*}$ following bad signals means less adverse selection.

Nevertheless, in the extension we considered where each agent can disclose up to $N$ signals, under generic parameter values $(u,v,\mathbf{f},\pi_0,\delta)$, all pure-strategy equilibria are strict. This is because when $N$ is finite, there is only a finite number of signal profiles that can be disclosed. If two signal profiles lead to the same expected payoff for the sender, then it leads to a polynomial of $\delta$. Since there is a finite number of signal profiles, there can be at most a finite number of $\delta$ for every $(u,v,\mathbf{f},\pi_0)$ under which the sender is indifferent. 

\newpage
\appendix
\section{Omitted Proofs}
\subsection{Proof of Lemma \ref{L0}}\label{secE}
Fix any $\theta$ and pure strategy $\sigma$,\footnote{The  proof applies to mixed strategies by defining $P^\sigma_\theta(\mathbf{h},\mathbf{h}') =
(1-\delta)\mathds{1}_{\mathbf{h}'=\mathbf{0}}
+ \delta\sum_{i=1}^l f_{\theta,i}\sigma_D(\mathbf{h},s_i)(\mathbf{h'})$.} and recall that $\sigma_D: H \times S \rightarrow H$ stands for the component of $\sigma$ that maps an agent's history and private signal to the set of signals disclosed to the next agent. 
For any $\mathbf{h},\mathbf{h}'\in H$, define the probability of moving from history $\mathbf{h}$ to history $\mathbf{h'}$ as
\begin{equation*}
P^\sigma_\theta(\mathbf{h},\mathbf{h}') =
(1-\delta)\mathds{1}_{\mathbf{h}'=\mathbf{0}}
+ \delta\sum_{i=1}^l f_{\theta,i}\mathds{1}_{\sigma_D(\mathbf{h},s_i)=\mathbf{h'}}.
\end{equation*}

The above equation defines a stochastic process on histories with transition probability function $P^\sigma_\theta$. The process is Markov, since each player's strategy can only depend on their disclosed history and signal, which is i.i.d. over time conditional on the state. Conditional on the selected history being successfully observed by the next generation, we can also define a transition probability $Q^\sigma_\theta(\mathbf{h},\mathbf{h'}) =
\sum_{i=1}^l f_{\theta,i}\mathds{1}_{\sigma_D(\mathbf{h},s_i)=\mathbf{h'}}$ and observe that it also defines a Markov process. Thus, we can write 
\begin{equation} \label{eqLem1transition}
    P^\sigma_\theta(\mathbf{h},\mathbf{h}') =
(1-\delta)\mathds{1}_{\mathbf{h}'=\mathbf{0}}
+ \delta Q^\sigma_\theta(\mathbf{h},\mathbf{h}').
\end{equation}

We show existence and uniqueness of a steady state distribution using the Banach fixed point theorem. Let $T^\sigma_\theta:\Delta(H) \rightarrow \Delta(H)$ be the transition operator implied by $P^\sigma_\theta$, so that $(T^\sigma_\theta(\mu))(\mathbf{h'}) = \sum_{\mathbf{h}\in H}\mu(\mathbf{h}) P^\sigma_\theta (\mathbf{h},\mathbf{h'})$. We claim that this is a contraction mapping when equipped with the total variation distance. 

Note that applying only the $Q^\sigma_\theta$ transition does not increase total variation distance. This follows because $Q^\sigma_\theta$ is a Markov transition probability. Formally, for any $\mu,\mu'\in\Delta(H)$,
\begin{align*}
\Big\|\sum_{\mathbf{h}}\mu(\mathbf{h})Q^\sigma_\theta(\mathbf{h},\cdot)-\sum_{\mathbf{h}}\mu'(\mathbf{h})Q^\sigma_\theta(\mathbf{h},\cdot)\Big\|_{\mathrm{TV}}
&=
\frac12\sum_{\mathbf{h}'\in H}\Big|\sum_{\mathbf{h}\in H}\big(\mu(\mathbf{h})-\mu'(\mathbf{h})\big)Q^\sigma_\theta(\mathbf{h},\mathbf{h}')\Big|\\
&\le
\frac12\sum_{\mathbf{h}'\in H}\sum_{\mathbf{h}\in H}\big|\mu(\mathbf{h})-\mu'(\mathbf{h})\big|\,Q^\sigma_\theta(\mathbf{h},\mathbf{h}')\\
&=
\frac12\sum_{\mathbf{h}\in H}\big|\mu(\mathbf{h})-\mu'(\mathbf{h})\big|\sum_{\mathbf{h}'\in H}Q^\sigma_\theta(\mathbf{h},\mathbf{h}')\\
&=
\frac12\sum_{\mathbf{h}\in H}\big|\mu(\mathbf{h})-\mu'(\mathbf{h})\big|\\
&=
\|\mu-\mu'\|_{\mathrm{TV}},
\end{align*}
where we used that $\sum_{\mathbf{h}'\in H}Q^\sigma_\theta(\mathbf{h},\mathbf{h}')=1$ for each $\mathbf{h}$.

By (\ref{eqLem1transition}) and the above inequality, we have that for any $\mu,\mu'\in\Delta(H)$,
\begin{equation*}
\|T^\sigma_\theta(\mu)-T^\sigma_\theta(\mu')\|_{\mathrm{TV}}
\le
\delta\,\|\mu-\mu'\|_{\mathrm{TV}}.
\end{equation*}
Since $\delta\in(0,1)$, $T^\sigma_\theta$ is a contraction on $(\Delta(H),\|\cdot\|_{\mathrm{TV}})$. By the Banach fixed point theorem, there exists a unique steady state distribution over histories $\mu_\theta\in\Delta(H)$.

\subsection{Proof of Proposition \ref{Prop1}}\label{secC}
Suppose agents disclose no signal other than $s_i$ and disclose up to $n$ realizations of $s_i$. Then conditional on state $\theta$, for every $m<n$, the probability of history $m \mathbf{1}_i$ is 
\begin{equation}\label{C.1}
(1-\delta) \delta^{m}  \cdot (1-\delta (1-f_{\theta,i}))^{-(m+1)} \cdot  f_{\theta,i}^{m}.
\end{equation}
Therefore, 
\begin{equation}\label{C.2}
\frac{\pi(m \mathbf{1}_i)}{1-\pi(m \mathbf{1}_i)}=
\frac{\pi_0}{1-\pi_0} \cdot
\Big(
\frac{1-\delta (1-\underline{f}_i) }{1-\delta (1-\overline{f}_i)}
\Big)^{m+1} \cdot 
L_i^m.
\end{equation}
If $L_i >1$, then the value of (\ref{C.2}) is strictly increasing in $m$ regardless of $\delta$. 
The probability of history $n \mathbf{1}_i$ is
\begin{equation}\label{C.3}
\sum_{k=n}^{\infty} (1-\delta) \delta^{k}  \cdot (1-\delta (1-f_{\theta,i}))^{-(k+1)} \cdot  f_{\theta,i}^{k}, 
\end{equation}
so 
\begin{align}\label{C.4}
\frac{\pi(n \mathbf{1}_i)}{1-\pi(n \mathbf{1}_i)}
&\ge 
\frac{\pi_0}{1-\pi_0} \cdot 
\min_{k \geq n} \Big\{
\frac{ (1-\delta)\,\delta^{k}
       \big(1-\delta (1-\overline{f}_i)\big)^{-(k+1)}
       \overline{f}_i^{\,k}
     }
     { (1-\delta)\,\delta^{k}
       \big(1-\delta (1-\underline{f}_i)\big)^{-(k+1)}
       \underline{f}_i^{\,k}
     } \Big\}
\notag \\[0.4em]
&\ge 
\frac{\pi_0}{1-\pi_0} \cdot 
\frac{ (1-\delta)\,\delta^{n}
       \big(1-\delta (1-\overline{f}_i)\big)^{-(n+1)}
       \overline{f}_i^{\,n}
     }
     { (1-\delta)\,\delta^{n}
       \big(1-\delta (1-\underline{f}_i)\big)^{-(n+1)}
       \underline{f}_i^{\,n}
     }
\notag \\[0.6em]
&=
\frac{\pi_0}{1-\pi_0} \cdot
\left(
\frac{1-\delta (1-\underline{f}_i)}
     {1-\delta (1-\overline{f}_i)}
\right)^{\!n+1}
L_i^{\,n}.
\end{align}
The right-hand-side of (\ref{C.4}) is strictly greater than the value of (\ref{C.2}) for any $m<n$. 
These calculations verified that regardless of $\delta \in (0,1)$, there exists an equilibrium where agents only disclose up to $n$ realizations of $s_i$. The resulting belief map $\pi: H \rightarrow [0,1]$ can be constructed so that $\pi(\mathbf{h})< \pi(\mathbf{0})$ for every $\mathbf{h} \in H$ that occurs with zero probability under the agents' equilibrium strategy.

\subsection{Proof of Proposition \ref{thm:MIcont}}\label{secF}
    For any $n \geq 2$, fix $\delta'<\delta$ so that $\delta', \delta \in (\delta_{n+1},\delta_{n}]$. We first show that $I(\theta;\mathbf{h}_{\delta'}) \leq I(\theta;\mathbf{h}_\delta)$. Observe that both $\mathbf{h}_\delta$ and $\mathbf{h}_{\delta'}$ involve disclosing signals $\{s_1,...,s_n\}$. 
    We refer to the sequence of signals after the last communication failure (including the ones concealed) as the \textit{complete history} of game. Since the expected length of the complete history in the steady state equals $\delta/(1-\delta)$, the expected number of signals disclosed is increasing in $\delta$. Hence, if $\delta'<\delta$, then the random variable $\mathbf{h}_{\delta'}$ can be represented as a garbling of the random variable $\mathbf{h}_\delta$, which implies that $I(\theta;\mathbf{h}_{\delta'}) \leq I(\theta;\mathbf{h}_\delta)$. Observe also that for such $\delta$, $I(\theta;\mathbf{h}_\delta)$ varies continuously since $\delta/(1-\delta)$ is continuous in $\delta$. 

    Fix any cutoff $\delta_n$. By definition, signals $\{s_1,...,s_n\}$ are disclosed in the most informative constant-threshold equilibrium when $\delta$ is slightly below $\delta_n$ and signals $\{s_1,...,s_{n-1}\}$ are disclosed in the most informative constant-threshold equilibrium when $\delta$ is slightly above $\delta_n$
    At the cutoff $\delta_n$, the incentive constraint to disclose signal $n$ is binding, so that by equation (\ref{3.8}) we have
\begin{equation*}
    L_n \cdot \frac{1-\delta_n (1-\underline{f_1}-...-\underline{f_n})}{1-\delta_n (1-\overline{f_1}-...-\overline{f_n})} = 1.
\end{equation*}
    But because of the binding constraint, disclosing an additional realization of $s_n$ does not change the likelihood ratio of the public belief. Hence, when $\delta=\delta_n$, we have $\pi(h_1,...,h_{n-1},0,...,0)=\pi(h_1,...,h_{n-1},h_{n},0,...,0)$ for any $h_n \in \mathbb{N}$ and $(h_1,...,h_{n-1}) \in \mathbb{N}^{n-1}$. Therefore, we have $I(\theta;\mathbf{h}^{(n)}_{\delta_n})=I(\theta;\mathbf{h}^{(n-1)}_{\delta_n})=I(\theta;\mathbf{h}_{\delta_n})$ and $I(\theta;\mathbf{h}_\delta)$ is continuous at $\delta=\delta_n$.

\subsection{Proof of Lemma \ref{L2}}\label{secB}
Suppose by way of contradiction that the realized history in period $t-1$ was $\mathbf{\hat{h}} \equiv (\hat{h}_1,...,\hat{h}_l)$, which is neither  $\mathbf{h^*}$ nor 
$\mathbf{h^*}-\mathbf{1}_j$. We obtain a contradiction in three cases:
\begin{enumerate}
    \item If $\hat{h}_j=0$, then $h_j^*-\hat{h}_j=1$, so 
    the gradualism property implies that $\sum_{i=1}^{j-1} \hat{h}_i \geq \sum_{i=1}^{j-1} h_i^*$. Since $\mathbf{\hat{h}} \neq \mathbf{h^*}-\mathbf{1}_j$, we have $(\hat{h}_1,...,\hat{h}_{j-1}) \neq (h_1^*,...,h_{j-1}^*)$. Hence, there exists $1 \leq i \leq j-1$ such that $\hat{h}_i > h_i^*$. This implies that after observing signal $s_j$ at history $\mathbf{\hat{h}}$, the agent discloses $s_j$ yet conceals at least one realization of $s_i$. When the belief map $\pi$ is monotone, disclosing $\mathbf{h^*}$ leads to a strictly lower belief relative to disclosing $\mathbf{h^*}-\mathbf{1}_j +\mathbf{1}_i$. The latter is feasible by construction. This contradicts the agent's incentive constraint. 
    \item If $\hat{h}_j=1$, then the construction of $\mathbf{h^*}$ implies that $\sum_{i=1}^{j-1} \hat{h}_i \geq \sum_{i=1}^{j-1} h_i^*$. Since $\mathbf{\hat{h}} \neq \mathbf{h^*}$, we have $(\hat{h}_1,...,\hat{h}_{j-1}) \neq (h_1^*,...,h_{j-1}^*)$, so there exists $1 \leq i \leq j-1$ such that $\hat{h}_i > h_i^*$. 
    Again, the monotonicity of the belief map $\pi$ implies that disclosing $\mathbf{h^*}$ leads to a strictly lower belief relative to disclosing $\mathbf{h^*}-\mathbf{1}_j +\mathbf{1}_i$. By construction, vector $\mathbf{h^*}-\mathbf{1}_j +\mathbf{1}_i$ is also feasible for the agent, so disclosing $\mathbf{h^*}$ violates his incentive constraint. 
    \item If $\hat{h}_j \geq 2$, then consider the sequence of agents' realized histories starting from the latest communication failure, which we normalize as period $0$. Let $\mathbf{h}(\tau)$ denote the realized history in period $\tau$. By definition, $\mathbf{h}(0)=(0,0,...,0)$, $\mathbf{h}(t-1)=\mathbf{\hat{h}}$, and $\mathbf{h}(t)=\mathbf{h^*}$. The gradualism property implies the existence of $\tau<t-1$ such that the $j$th entry of $\mathbf{h}(\tau)$ is $1$. Let $\mathbf{h}(\tau) \equiv \mathbf{h^{**}} \equiv (h^{**}_1,...,h^{**}_{j-1},1,0,...,0)$. The definition of vector $\mathbf{h^*}$ implies that at least one of the following two cases is true. The first case is $\sum_{i=1}^{j-1} h^{**}_i = \sum_{i=1}^{j-1} h_i^*$, which by the construction of history $\mathbf{h^*}$, implies that 
    \begin{equation*}
    \sum_{i=1}^k h^{**}_i \geq \sum_{i=1}^k h^*_i \textrm{ for every } i \in \{1,2,...,j-1\}.
\end{equation*}
The second case is  $\sum_{i=1}^{j-1} h^{**}_i > \sum_{i=1}^{j-1} h^*_i$, in which case there exists $\mathbf{h} \equiv (h_1,...,h_{j-1},1,0,...,0)$ that is weakly less than $\mathbf{h^{**}}$ such that 
    \begin{equation*}
    \sum_{i=1}^k h_i \geq \sum_{i=1}^k h_i^* \textrm{ for every } i \in \{1,2,...,j-1\}.
\end{equation*}
The monotonicity of the belief map $\pi$ then implies the existence of $\mathbf{h} \leq \mathbf{h^{**}}$ 
such that $\mathbf{h}$ and $\mathbf{h^*}$ share the same length and moreover,
$\pi(\mathbf{h}) \geq \pi(\mathbf{h^{*}})$. Since history $\mathbf{h^{**}}$ occurs with positive probability in equilibrium, we have $\pi(\mathbf{h^{**}}) \geq \pi(\mathbf{h})$ for every $\mathbf{h} \leq \mathbf{h^{**}}$. Therefore, $\pi(\mathbf{h^{**}}) \geq \pi(\mathbf{h^*})$. However $\mathbf{h}(t)=\mathbf{h^*}$ and $\mathbf{h}(\tau)=\mathbf{h^{**}}$ with $t> \tau$ along a realized path of histories, and a history in between them $\mathbf{h}(t-1)=\mathbf{\hat{h}} \neq \mathbf{h^{**}}$. The conclusion that 
 $\pi(\mathbf{h^{**}}) \geq \pi(\mathbf{h^*})$ contradicts the improvement principle we established earlier, which requires that  $\pi(\mathbf{h^{**}}) < \pi(\mathbf{h^*})$. 
\end{enumerate} 

\paragraph{Remark:} The proof of Lemma \ref{L2} extends to the model discussed in Section \ref{sechistorydependent} where communication success rate weakly decreases as the disclosed history gets longer. To see this, let us revisit the three cases considered in the proof. 
The first two cases only involve comparison between beliefs at two histories of the same length, namely, $\mathbf{h^*}$ and $\mathbf{h^*}-\mathbf{1}_j +\mathbf{1}_i$. Since the communication success rate depends only on the length of the disclosed history, an agent prefers to disclose 
$\mathbf{h^*}$ to $\mathbf{h^*}-\mathbf{1}_j +\mathbf{1}_i$
if and only if $\pi(\mathbf{h^*}) \geq \pi(\mathbf{h^*}-\mathbf{1}_j +\mathbf{1}_i)$. In the third case, recall the construction of history $\mathbf{h^{**}}$ as well as the existence of history $\mathbf{h}$ such that 
$\mathbf{h} \leq \mathbf{h^{**}}$, $\mathbf{h}$ and $\mathbf{h^*}$ share the same length, and 
$\pi(\mathbf{h}) \geq \pi(\mathbf{h^{*}})$. Therefore, the expected payoff from disclosing $\mathbf{h}$ is weakly greater than the expected payoff from disclosing $\mathbf{h^*}$, as the communication success rates are the same yet disclosing $\mathbf{h}$ leads to a weakly higher belief. This leads to a contradiction since disclosing $\mathbf{h}$ is available when the history is $\mathbf{h^{**}}$ yet history $\mathbf{h^*}$ is reached after $\mathbf{h^{**}}$ before the next time communication breaks down. 

\subsection{Proof of Lemma \ref{L1}}\label{secA}
Let $\mathbf{h'} \equiv (h_1',...,h_l')$ denote his immediate predecessor's history and suppose by way of contradiction that $\mathbf{h'} \notin \{ \mathbf{h^*}
, \mathbf{h^*} -\mathbf{1}_j \}$. We consider three cases separately:
\begin{enumerate}
    \item If $h_j'=0$, then in order for history to evolve from $\mathbf{h'}$ to $\mathbf{h^*}$ in one period, 
    the current-period private signal must be $s_j$ and $s_j$ needs to be disclosed after receiving it at history $\mathbf{h'}$. Since $\sigma$ is a threshold strategy, the fact that $s_j$ is disclosed 
    after receiving private signal $s_j$ at history
    $\mathbf{h'}$ implies that existing signals $s_1,...,s_{j-1}$ are all disclosed. This implies that $\mathbf{h'}=(h_1^*,...,h_{j-1}^*,0,...,0)$, which contradicts our hypothesis that $\mathbf{h'} \neq \mathbf{h^*}-\mathbf{1}_j$. 
    \item If $h_j'=1$, then the construction of $\mathbf{h^*}$ requires that $\sum_{i=1}^{j-1} h_i' \geq \sum_{i=1}^{j-1} h_i^*$ and our hypothesis that $\mathbf{h'} \neq \mathbf{h^*}$ implies that $(h_1',...,h_{j-1}') \neq (h_1^*,...,h_{j-1}^*)$. The two together imply that there exists $i \in \{1,2,...,j-1\}$ such that $h_i' > h_i^*$. This implies that after receiving some signal at history $\mathbf{h'}$, an agent will disclose signal $s_j$ while concealing at least one realized $s_i$ for some $i<j$. This contradicts the hypothesis that agents use a threshold strategy. 
    \item If $h_j' \geq 2$, then after receiving some signal at $\mathbf{h'}$, agents will conceal at least one realized $s_j$ while reveal at least one $s_j$. This contradicts the hypothesis that $\sigma$ is a threshold strategy. 
\end{enumerate}

\subsection{Proof of Proposition \ref{Prop4}}\label{secG}
If the sender uses Strategy $k$, then the receiver will only observe 
$\mathbf{h} \equiv (k+1) \mathbf{1}_1 + \mathbf{1}_2$ and 
$j\mathbf{1}_1$ for every $j \neq k+1$ with positive probability. The receiver's posterior beliefs,
given by the mapping $\pi: H \rightarrow [0,1]$, are pinned down by Bayes rule following these positive probability events. 
Therefore, in order to verify that Strategy $k$ is part of a strict equilibrium, we only need to verify that the following three conditions are satisfied:
\begin{enumerate}
    \item $\pi(j \mathbf{1}_1)$ is strictly increasing in $j$ for every $j \neq k+1$.
    \item $\pi((k+2)\mathbf{1}_1)>\pi( \mathbf{h} )$.
    \item $\pi( \mathbf{h} ) > \pi(k \mathbf{1}_1)$.
\end{enumerate}
The first two conditions are satisfied for all $k \in \mathbb{N}$, as
implied by expressions (\ref{3.2}) to (\ref{3.5})
 in the proof of Theorem \ref{Theorem2}. In what follows, we show that for every $\delta \in (0,1)$, there exists a large enough $k$ such that the third condition is also satisfied. Applying (\ref{3.2}) and (\ref{3.4}), we know that if the sender uses Strategy $k$ in the one-shot game, then conditional on the state being $\theta$, the receiver will observe $k \mathbf{1}_1$ with probability
\begin{equation*}
(1-\delta) \delta^k \Big(1-\delta (1-f_{\theta,1})\Big)^{-(k+1)} f_{\theta,1}^k
+ (1-\delta) \delta^{k+1} \Big(
1-\delta (1-f_{\theta,1}-f_{\theta,2})
\Big)^{-(k+2)} f_{\theta,1}^{k+1},
\end{equation*}
where the first term is the probability that the sender's observes $k$ signal $s_1$ and an unknown number of signals from $s_2$ to $s_l$, and the second term is the probability that the sender observes $k+1$ signal $s_1$, no signal $s_2$, as well as an unknown number of signals from $s_3$ to $s_l$. Under the same sender-strategy and conditional on state $\theta$,
the receiver will observe $\mathbf{h}$ with probability
\begin{equation*}
(1-\delta) \delta^{k+1} \Big(1-\delta (1-f_{\theta,1})\Big)^{-(k+2)} f_{\theta,1}^{k+1}
- (1-\delta) \delta^{k+1} \Big(
1-\delta (1-f_{\theta,1}-f_{\theta,2})
\Big)^{-(k+2)} f_{\theta,1}^{k+1},
\end{equation*}
where the first term is the probability that the sender observes $k+1$ signal $s_1$ and an unknown number of signals from $s_2$ to $s_l$ and the second term coincides with the second term in the previous expression except with the opposite sign.
According to Bayes rule, we know that $\pi( \mathbf{h} ) > \pi(k \mathbf{1}_1)$ if and only if the value of 
\begin{equation}\label{6.1}
\frac{\Big(1-\delta (1-f_{\theta,1})\Big)^{-(k+1)} + \delta f_{\theta,1} \Big(
1-\delta (1-f_{\theta,1}-f_{\theta,2})
\Big)^{-(k+2)}}
{
f_{\theta,1} \Big(1-\delta (1-f_{\theta,1})\Big)^{-(k+2)} - f_{\theta,1} \Big(
1-\delta (1-f_{\theta,1}-f_{\theta,2}) \Big)^{-(k+2)}
}
\end{equation}
strictly decreases once $\theta$ increases from $\underline{\theta}$ to $\overline{\theta}$. This is the case if and only if 
\begin{equation}\label{6.2}
\frac{ \Big( \overline{X}^{-(k+1)} + \delta \overline{f_1} \cdot \overline{Y}^{-(k+2)} \Big) \cdot \Big(\underline{X}^{-(k+2)} - \underline{Y}^{-(k+2)} \Big) }{\Big( \underline{X}^{-(k+1)} + \delta \underline{f_1} \cdot \underline{Y}^{-(k+2)} \Big) \cdot \Big( \overline{X}^{-(k+2)} - \overline{Y}^{-(k+2)} \Big)} < \frac{\overline{f_1}}{\underline{f_1}}=L_1,
\end{equation}
where
\begin{equation*}
\overline{X} \equiv  1-\delta (1-\overline{f_1}), \quad  \underline{X} \equiv  1-\delta (1-\underline{f_1}) 
\end{equation*}
\begin{equation*}
\overline{Y} \equiv  1-\delta (1-\overline{f_1}-\overline{f_2}), \textrm{ and } \underline{Y} \equiv  1-\delta (1-\underline{f_1}-\underline{f_2}).  
\end{equation*}
Let 
\begin{equation*}
\overline{R} \equiv \overline{X}/\overline{Y} \textrm{ and } \underline{R} \equiv \underline{X}/\underline{Y}.
\end{equation*}
Inequality (\ref{6.2}) can then be rewritten as
\begin{equation}\label{6.3}
\frac{ ( \overline{X} + \delta \overline{f_1} \cdot \overline{R}^{k+2} ) (1-\underline{R}^{k+2}) }{  ( \underline{X} + \delta \underline{f_1} \cdot \underline{R}^{k+2} ) (1-\overline{R}^{k+2}) } < L_1
\end{equation}
For every $\delta \in (0,1)$, we know that $\overline{X}<\overline{Y}$ and $\underline{X} < \underline{Y}$, which implies that $0< \underline{R} < \overline{R} <1$. Therefore, fix any $\delta \in (0,1)$, both $\overline{R}^{k+2}$ and $\underline{R}^{k+2}$ converge to $0$ as $k \rightarrow +\infty$.
This suggests that left-hand-side of (\ref{6.3}) converges to 
$\overline{X}/\underline{X}$ as $k \rightarrow +\infty$. Since $L_1 > \overline{X}/\underline{X}$ for any $\delta \in (0,1)$, we know that for any $\delta \in (0,1)$, there exists a large enough $k$ under which inequality (\ref{6.3}) is satisfied. This verifies that no matter how close $\delta$ is to $1$,
there exists a strict monotone equilibrium  under which signal $s_2$ is disclosed with positive probability.

\subsection{Proof of Proposition \ref{Prop3}}\label{secD}
Suppose by way of contradiction that for every $\delta$ close enough to $1$, there exists a strict equilibrium in the one-shot disclosure game where the sender uses a threshold strategy $\sigma_D$ under which there exists $j \geq 2$ such that 
the $j$th entry of
$\sigma_D(\mathbf{h})$
is strictly positive for some  $\mathbf{h} \in H$.

Fix any such equilibrium and without loss of generality, let $j \in \{2,...,l\}$ denote the largest integer such that there exists $\mathbf{h}$ such that the $j$th entry of $\sigma_D(\mathbf{h})$ is not zero. Let 
$\mathbf{h^*} \equiv (h_1^*,...,h_l^*) \in H$ denote a history at which $s_j$ is disclosed. 
Since $\sigma_D$ is a threshold strategy, we have $\sigma_D(\mathbf{h^*})=\mathbf{h^*}[j]$.
For every $k \in \mathbb{N}$, let $\mathbf{h^k} \equiv \mathbf{h^*}[j] + k \mathbf{1}_j$.  
\begin{Lemma}\label{L3}
For every $\mathbf{h} \in H $ such that there exists an integer $k \geq 1$ that satisfies
$\mathbf{h}[j]= \mathbf{h^k}$,
we have $\sigma_D(\mathbf{h})=\mathbf{h^k}[j]$. 
\end{Lemma}
\begin{proof}[Proof of Lemma \ref{L3}:]
Since $j$ is the highest signal index that can be disclosed, the equilibrium is strict, and 
$\sigma_D$ is a threshold strategy, we know that $\sigma_D(\mathbf{h^*}[j]) = \mathbf{h^*}[j]$ and
\begin{equation}\label{D.1}
\pi(\mathbf{h^*}[j]) > \pi(\mathbf{h}) \textrm{ for every } \mathbf{h} < \mathbf{h^*}[j]. 
\end{equation}
Once the sender observes $\mathbf{h^1}$, we from (\ref{D.1}) that $\sigma_D (\mathbf{h^1})$ cannot be anything strictly less than $\mathbf{h^*}[j]$ since disclosing $\mathbf{h^*}[j]$ is available. Since $\sigma_D$ is a threshold strategy, all available signal $s_j$ must be disclosed as long as one $s_j$ is disclosed, which implies that $\sigma_D(\mathbf{h^1})=\mathbf{h^1}[j]$. Using the same logic, one can show that for every $k \in \mathbb{N}$, the hypothesis that 
$\sigma_D(\mathbf{h^k})=\mathbf{h^k}[j]$ implies that
$\sigma_D(\mathbf{h^{k+1}})=\mathbf{h^{k+1}}[j]$. This induction argument implies that
$\sigma_D(\mathbf{h^k})=\mathbf{h^k}[j]$ for every $k \geq 1$. The definition of $j$ implies that signals that do not belong to $\{s_{1},...,s_j\}$ will never be disclosed under $\sigma_D$, which implies that $\sigma_D(\mathbf{h})=\mathbf{h^k}[j]$ for every 
$\mathbf{h}$  that satisfies
$\mathbf{h}[j]= \mathbf{h^k}$.
\end{proof}
Lemma \ref{L3} suggests that if the receiver observes history $\mathbf{h^k}$, then she infers that the sender has disclosed all the available $\{s_1,...,s_j\}$ and has concealed all the other signals.
Recall our hypothesis that $j \geq 2$.
Applying the calculations in the proof of Theorem \ref{Theorem2} to $\pi(\mathbf{h^1})$ and $\pi(\mathbf{h^2})$, we know that 
there exists $\delta^* \in (0,1)$ such that 
$\pi(\mathbf{h^1}) > \pi(\mathbf{h^2})$ 
as long as $\delta > \delta^*$. This contradicts the sender's incentive to disclose $\mathbf{h^2}$ at $\mathbf{h^2}$ which requires that 
 $\pi(\mathbf{h^1}) \leq \pi(\mathbf{h^2})$.

\newpage
\bibliographystyle{ecta}
\bibliography{ref}

\end{document}